\documentclass[runningheads]{llncs}
\usepackage[T1]{fontenc}
\usepackage[]{graphicx}
\usepackage{array}
\usepackage{hyperref}
\usepackage{color}

\usepackage{amsmath}
\usepackage{amssymb}
\usepackage{mathtools}
\usepackage{algorithm}
\usepackage{algorithmic}
\usepackage[switch]{lineno}
\usepackage{subcaption}
\usepackage{csquotes}

\usepackage{cite} % Style requires multiple citations be sorted

\usepackage{pgfplots}
\usepackage{mathtools}

\usepackage{todonotes}
\usepackage{xcolor}

\newcommand{\citet}[1]{\citeauthor{#1}~\shortcite{#1}}

\usepackage{style}

\allowdisplaybreaks

\begin{document}
\title{Contested Route Planning}
% \author{
% Jakub \v{C}ern\'{y}\inst{1}\orcidID{0000-0002-9407-3782} \and
% Chun Kai Ling\inst{1}\orcidID{0000-0001-9405-8972} \and
% Darshan Chakrabarti\inst{1}\orcidID{0000-0002-3907-044X} \and 
% Jingwen Zhang\inst{2}\orcidID{XXXX-XXXX-XXXX-XXXX} \and
% Gabriele Farina\inst{3}\orcidID{0000-0002-3976-0061} \and
% Christian Kroer\inst{1}\orcidID{0000-0002-9009-8683} \and
% Garud Iyengar\inst{1}\orcidID{0000-0001-6546-4154}
% }

% \author{Submission 34}

\author{
Jakub \v{C}ern\'{y} 
\and
Garud Iyengar
\and
Christian Kroer
}
\authorrunning{J. \v{C}ern\'{y} et al.}
\institute{Columbia University, New York, NY 10027, USA\\
\email{\{jakub.cerny, christian.kroer\}@columbia.edu}, \email{garud@ieor.columbia.edu}}

\maketitle              % typeset the header of the contribution
\begin{abstract}
We consider the problem of routing for logistics purposes, in a contested environment where an adversary attempts to disrupt the vehicle along the chosen route. We construct a game-theoretic model that captures the problem of optimal routing in such an environment. While basic robust deterministic routing plans are already challenging to devise, they tend to be predictable, which can limit their effectiveness. By introducing calculated randomness via modeling the route planning process as a two-player zero-sum game, we compute immediately deployable plans that are diversified and harder to anticipate. Although solving the game exactly is intractable in theory, our use of the double-oracle framework enables us to achieve computation times on the order of seconds, making the approach operationally viable. In particular, the framework is modular enough to accommodate specialized routing algorithms as oracles. We evaluate our method on real-world scenarios, showing that it scales effectively to realistic problem sizes and significantly benefits from explicitly modeling the adversary's capabilities, as demonstrated through ablation studies and comparisons with baseline approaches.
\end{abstract}

\section{Introduction}

Pathfinding is the process of determining the most efficient route between two points, a problem with historical roots in graph theory and operations research, dating back to the mid-20th century. Early algorithms like Dijkstra’s~\cite{dijkstra2022note} and A$^*$~\cite{hart1968formal} laid the foundation for modern approaches. Today, route planning is crucial in applications like GPS navigation, robotics, logistics, and video games. In services like Google Maps, algorithms account for factors such as road distances, speed limits, real-time traffic, road closures, and user preferences to compute the fastest route. These systems often use a combination of heuristics, live data, and precomputed route segments to ensure both speed and accuracy in routing.

While traditional route planning methods perform well in environments where risks are passive or stochastic, they may fall short when navigating through territories with active adversaries. In this paper, we explore scenarios where route planning must account for the presence of an opponent intent on disrupting movement -- whether for transferring people or goods, or as part of larger logistical operations. Attacks on supply lines have been widely documented in real military contexts and are often considered more strategically effective than direct confrontations~\cite{mcmahon2017maritime}. In a broader geopolitical context, major powers like the US, China, and the EU actively seek to diversify their supply chains to reduce vulnerability to potential conflict-related disruptions~\cite{cowen2010geography}. Beyond conflict preparedness, incorporating adversarial considerations into route planning serves as a form of stress testing -- highlighting fragilities that might otherwise remain hidden. Historical cases such as the repeated Allied bombing of the bridge on the River Kwai during World War II~\cite{warfare_kwai}, or more recently, the 2022 Ukrainian strikes on the Antonivskyi Bridge near Kherson~\cite{guardian_kherson} underscore how the loss of a single critical link can paralyze military logistics and alter the strategic landscape. Similarly, the 2021 EverGreen incident in the Suez Canal is a stark reminder of the vulnerabilities posed by overreliance on single points of failure \cite{lee2021suez}.

The introduction of an adversary naturally calls for a game-theoretic framework to model the strategic interaction between the route planner and the interdictor. Without explicitly modeling this interaction, route planning solutions risk being highly exploitable. First, they fail to anticipate the adversary's adaptive behavior -- how an informed opponent might change their interdiction strategy in response to the selected route. Second, deterministic route choices become predictable and easy to target when used repeatedly, even if they are adversary-aware. To address these issues, we introduce \textit{contested route planning} as a game-theoretic model involving two players, customarily referred to as \textit{Blue} and \textit{Red}. Blue represents the route planner seeking to optimize their route choices, while Red is the adversary attempting to disrupt or degrade Blue’s performance. This interaction is framed as a two-player zero-sum one-shot game, where strategies are chosen simultaneously. The solution is characterized by the Nash equilibrium, typically involving randomized strategies. Such game-theoretic solutions improve over non-strategic models by making the chosen routes both more robust to dynamic adversaries and less predictable over time.

While the presence of adversaries in route planning is not new~\cite{smith2020survey}, our model departs from prior work in several key ways. First, it does not assume a fixed behavior or limited rationality for the adversary. Instead, Red is modeled as a fully adaptive opponent capable of selecting responses that most effectively degrade Blue’s performance. Second, Red is modeled as making discrete, localized interruptions, such as destruction of roads, bridges, or convoys, rather than continuously distributing disruptions across the entire network. These requirements and the model itself were developed in cooperation with the U.S. Office of Naval Research, with the primary goal of providing a practically useful tool. In particular, our models and algorithms are designed to run efficiently on low-power devices (including handhelds), without relying on complex optimization software, and at scales relevant to real-world operations. To that end, the formulation is intentionally flexible, allowing integration with existing planners that incorporate vehicle-specific routing constraints.

A primary intended application of our model is tactical-level contested route (re)planning, supporting decision-making for individual vehicles (in the military context, this means operating at the battalion level and below), where human planners may struggle to balance concentration of movement with operational effectiveness. Our model is tailored for dynamic environments where routes must be rapidly reevaluated in response to evolving conditions on the ground, such as the sudden loss of a bridge or tunnel, or the emergence of hostile presence that increases the risk associated with certain paths. It also addresses the typical pitfalls of manual replanning, including high communication overhead for coordinating alternate courses of action, individual cognitive biases around risk, and difficulties in synchronizing updates across units.

Beyond lower-level planning, our model is equally applicable to higher-level logistical coordination, where multiple vehicles must execute deliveries across dispersed locations. Existing models that attempt to account for adversarial disruption at this scale often fail to capture realistic complexity or cannot scale when tasked with modeling both strategic attacks (e.g., on depots or large convoys) and fine-grained route choices~\cite{vcerny2024contested}. A practical solution to this problem is decomposition: the higher-level delivery plan is computed first, possibly using a coarser adversarial abstraction, and then the resulting individual routing tasks are delegated to our model. This layered approach allows for scalability while preserving adversarial robustness at the operational level.

The computation of a Nash equilibrium presents significant challenges. Both the number of possible Blue routing plans and the number of Red interdiction plans is exponential. Thus, explicitly specifying and solving the problem as a zero-sum bimatrix game is not practical. Our main contributions are as follows:

% \vspace{-.2cm}
\begin{itemize}
    \item We formally propose the framework of contested route planning games, a novel variant of pathfinding that accounts for Red's capabilities. Our min-max formulation \textit{explicitly} models Red actively seeking to thwart Blue, via localized attacks. We show that an optimal strategy exists for both players via von Neumann's minimax theorem. 
    
    \item We prove that computing the equilibrium as well as Red's best response, are \NP{}-hard problems. Nonetheless, the best responses of Blue (respectively, Red) to a fixed randomized Red (resp., Blue) strategy can be formulated as standard computer science problems (shortest path and knapsack problems) enabling the use of established algorithms.
    
    \item We propose solving the game via a double oracle method, utilizing the best-response formulations. We demonstrate scalability via experiments. 

    \item We conduct experiments on scenarios inspired by real-world conditions, observing that (i) realistically sized problem instances can be solved within seconds to a few minutes, (ii) both na\"ive and robust deterministic strategies for Blue are highly exploitable, and (iii) underestimating Red’s capabilities can in some cases result in a significant loss of performance.
\end{itemize}

\section{Related Work}

We review related work in route planning under uncertainty or adversarial conditions, while referring readers to~\cite{goldberg2005computing, cherkassky1996shortest, bast2016route, delling2009engineering} for classical shortest-path algorithms, and to~\cite{pfohl1998logistics, ghiani2004introduction, daganzo2005logistics} for traditional logistics models.

One closely related area is the class of vehicle routing problems (VRPs), which focus on optimizing routes for delivery or service vehicles. While the literature on robust VRPs is extensive, most models incorporate uncertainty in the form of fluctuating demands, travel times, time windows, or customer availability, rather than explicitly modeling adversarial behavior~\cite{ordonez2010robust, agra2013robust, sungur2008robust, lee2012robust}. When adversarial elements are introduced, they are typically captured via non-adaptive probabilistic models~\cite{blom2020inventory, alotaibi2018unmanned}. A particularly relevant subset of this literature involves bi-level formulations in ambush avoidance and hazardous materials (hazmat) routing~\cite{ruckle1976ambushing, salani2010ambush, boidot2014complete}. However, these works generally assume simpler or more constrained models of the adversary (Red), in contrast to the combinatorial action space we consider. Related approaches in routing interdiction then adopt deterministic formulations of Blue's decision space~\cite{church2004identifying, sadati2020r, bidgoli2018arc}.

A further line of research is network interdiction, which investigates strategies for removing or degrading arcs in a network to hinder an adversary’s traversal. Originally studied in~\cite{wollmer1964removing}, this field has expanded to include objectives such as increasing an evader’s shortest path to an exit or reducing the maximum flow between key nodes~\cite{washburn1995two, smith2020survey, smith2008algorithms, israeli2002shortest}. A comprehensive survey is provided in~\cite{smith2020survey}. Of particular interest to our work are simultaneous shortest path interdiction models that frame the interaction as a game and seek equilibria over competing strategies~\cite{washburn1995two, goldberg2017non}. These models, however, differ in several ways, often assuming multiple sources and destinations or different models of Red's capabilities.

Pursuit–evasion games form another related well-established body of literature, modeling interactions where a pursuer aims to detect or capture a moving evader, often under time constraints. These come under names such as Cops and Robbers or Games of Pursuit, across both continuous and discrete domains with varying assumptions on observability and movement~\cite{isaacs1999differential, friedman2013differential, weintraub2020introduction, bopardikar2008discrete, bonato2011game, parsons2006pursuit}. While mathematically rich, these models are generally computationally intractable outside of small or highly structured environments. Conceptually, they differ from our setting because the evader usually selects among multiple exits while the pursuer physically tracks the evader, rather than pre-positioning interdictions.

Security games represent another relevant framework, often involving mobile or stationary defenders patrolling targets or areas subject to attack. These models have led to real-world applications~\cite{jain2013security, pita2008armor, shieh2012protect, an2017stackelberg}, and a range of algorithmic advances have been made to manage large but structured action spaces, including path patrol strategies~\cite{korzhyk2010complexity}. However, in most security game formulations, the moving player's goal is to protect a region or set of targets over time, rather than to traverse a network from a source to a destination~\cite{cerny2024layered, jain2011double, zhang2017optimal}.

Finally, contested logistics, a domain especially relevant to military operations, considers the disruption of supply chains by adversaries. Many of these models simplify Red’s behavior, either assuming fixed stochastic behavior or limited observability~\cite{barahona2007inventory, salmeron2009stochastic, hill2009overview}. While some formulations use bi-level optimization, the resulting strategies are typically deterministic or may not scale to the size of settings we aim to study~\cite{bell2015military, vcerny2024contested}.

\section{Model}

The contested route planning game is played on a finite directed physical graph $\mathsf{G} = (\mathsf{V}, \mathsf{E})$, where the nodes $\mathsf{V}$ represent geographic locations, such as cities, towns, provinces, and other relevant points, and the edges $\mathsf{E}$ represent the traversable connections between them. These connections may include roads, railways, waterways, or other transport corridors, depending on the context and mode of travel. Each edge is associated with three key values: an interdiction cost for Red, representing the resources required to disrupt the edge; an interdiction penalty for Blue, quantifying the consequence of traversing an interdicted edge (often related to metrics such as the \enquote{probability to kill} used in military contexts); and a time or distance cost reflecting the effort required for Blue to traverse the edge under normal, non-interdicted conditions. In addition, $\mathsf{V}$ contains two specific locations, a start node and a target release node, that Blue is tasked with navigating between.

The game is played simultaneously. Red selects a subset of edges to interdict, constrained by a fixed budget. Blue selects a path from the start to the release node and incurs penalties associated with all interdicted edges traversed. Blue’s objective is to minimize the total accumulated penalty en route, while Red aims to maximize it. The game is zero-sum, i.e., Red's goal is to maximize the penalty.

\subsection{Blue's strategy space}

Let us denote Blue’s start node in the physical graph $\mathsf{G}$ as $v_s \in \mathsf{V}$ and the release point as $v_r \in \mathsf{V}$. Blue’s action space consists of all feasible paths in $\mathsf{G}$ connecting $v_s$ to $v_r$. When $\mathsf{G}$ is acyclic\footnote{This may happen, for example, if $\mathsf{G}$ represents a time-expanded graph.}, the path selection problem can be encoded as the solution to the following feasibility MILP:
\begin{equation}
    \label{blue-flows}
    \tag{\emph{F}}
    \begin{aligned}
    % \max~&h(f) \\
        \sum_{e\in \pedges^{-}(v_s)}f(e) &= \sum_{e\in \pedges^{+}(v_r)}f(e) = 1 \\
        \sum_{e\in \pedges^-(v)}f(e) &= \sum_{e\in \pedges^+(v)}f(e)&&\forall v\in\pvertices\backslash\{v_s,v_r\}\\
        f(e)&\in\{0,1\}.&&\forall e\in\pedges %\tag{\lambda}
    \end{aligned}
\end{equation}
We denote the space of all feasible solutions as \ref{blue-flows} and a single solution as $f\in\ref{blue-flows}$. In cases where the graph contains cycles, additional care is required to eliminate closed loops that are not part of the valid path from $v_s$ to $v_r$, but are nonetheless included in the feasible solution space of the formulation. One common way, which we also employ, is to introduce an objective that penalizes unnecessary movement, thereby discouraging such extra cycles. In the remainder of the text we therefore assume that $f$ represents a single valid path. Note that the formulation can also be extended to incorporate further constraints, such as prohibiting certain edge combinations, to account for operational limitations like turning radius or terrain restrictions for different types of vehicles.

\subsection{Red's strategy space}

The strategy space of Red is simpler than that of Blue. The model follows the structure of classic network interdiction problems, where Red selects a subset of edges in the physical graph $\pgraph$ to interdict, subject to a budget $B \geq 0$ and an interdiction cost function $C: \pedges \to \realp$. We assume that once an edge is interdicted, it remains disrupted for the entire duration of the game. Red’s set of feasible actions is described by the solutions to the following MILP:
\begin{equation}
    \label{red-interdiction}
    \tag{\emph{Y}}
    \begin{aligned}
    B &\geq \sum_{e\in \pedges}C(e)y(e),\qquad 
    y(e)\in\{0,1\}.&&\forall e\in\pedges
        % \bigg\{ y \in \{ 0, 1 \}^{|\pedges|} \: \big| \: B \geq \sum_{e\in \pedges}C(e)y(e) \bigg\}.
    \end{aligned}
\end{equation}
As with Blue’s action space, we refer to Red’s action space as the set of feasible solutions to \ref{red-interdiction}, and denote a specific interdiction plan as $y \in \ref{red-interdiction}$. Again, this formulation provides a compact and easily extendable way to represent Red’s choices within the game.

\subsection{Utility}

We define two edge-specific penalty functions: a time or distance penalty $T: \pedges \to \realp$, representing the cost of traversing an edge under normal conditions, and an interdiction penalty $P: \pedges \to \realp$, capturing the additional penalty incurred by Blue when traversing an edge interdicted by Red. These penalties can be rescaled arbitrarily to reflect their relative importance or trade-off ratio in the utility calculation. Given $f\in\ref{blue-flows}$ representing Blue’s selected path and $y\in\ref{red-interdiction}$ indicating Red’s interdiction plan, the utility for Blue is defined as:
\begin{equation}
\label{utility}
\tag{\emph{u}}
u(f,y) = \sum_{e \in \pedges} T(e) f(e) + P(e) f(e) y(e).
\end{equation}

This utility reflects the total cost incurred by Blue: the sum of the baseline time or distance penalties, plus any additional penalties from traversing interdicted edges. Blue aims to minimize this total cost, while Red, since this is a zero-sum game, seeks to maximize it.

Importantly, this utility function also serves to implicitly eliminate cycles in the MILP formulation \ref{blue-flows} of Blue’s path space. Since each edge carries a nonzero traversal cost, solutions that include extraneous cycles (which would only add cost) are automatically disfavored, ensuring that the selected solution (a single path only) is both feasible and efficient.

\section{Solution}

Our goal is to compute a Nash equilibrium between Blue’s routing strategies and Red’s interdiction strategies. Let $\Delta_b$ and $\Delta_r$ denote the probability simplices over the feasible action sets defined by \ref{blue-flows} and \ref{red-interdiction}, respectively. A mixed strategy $x_i \in \Delta_i$ for player $i \in {b, r}$ assigns a probability $x_i(p_i)$ to each pure strategy $p_i$ in their action space. The equilibrium problem can be formulated as the following bilinear saddle-point optimization:
\begin{align*}
    &\min_{x_b \in \Delta_b} \max_{x_r \in \Delta_r} \mathbb{E}_{f \sim x_b, y \sim x_r} \left[ u (f, y) \right] \label{eq:general-min-max}
    = &\min_{x_b \in \Delta_b} \max_{x_r \in \Delta_r} \sum_{f\in\ref{blue-flows}} \sum_{y\in\ref{red-interdiction}} x_b(f) \cdot x_r(y) \cdot u(f, y). %\nonumber
\end{align*}
Because the physical graph $\pgraph$ is finite, both Blue and Red have finite sets of pure strategies. Their respective mixed strategy spaces are convex and compact, and the expected utility function is bilinear -- convex in Blue’s strategy and concave in Red’s. These properties ensure, via the minimax theorem~\cite{v1928theorie}, that a Nash equilibrium exists and the game has a well-defined value. Nonetheless, solving for this equilibrium is computationally demanding due to the exponential number of pure strategies in each player’s action space.

\begin{figure}[t]
    \centering

\begin{tikzpicture}[minimum size=10mm, node distance=1.8cm, >=stealth, thick]
    % Nodes
    \node[circle, draw, fill=white] (vs) {$v_s$};
    \node[right=0.7cm of vs] (dots0) {$\dots$};
    \node[circle, draw, fill=white, right=0.7cm of dots0] (vi) {$v_{i-1}$};
    \node[right=0.7cm of vi] (dots1) {$\dots$};
    \node[circle, draw, fill=white, right=0.7cm of dots1] (vip1) {$v_{i}$};
    \node[right=0.7cm of vip1] (dots2) {$\dots$};
    \node[circle, draw, fill=white, right=0.7cm of dots2] (vn) {$v_r$};

    % Edges with arrows
    \draw[->] (vs) -- (dots0);
    \draw[->] (dots0) -- (vi);
    \draw[-] (vi) -- (dots1);
    \draw[->] (dots1) -- (vip1);
    \draw[->] (vip1) -- (dots2);
    \draw[->] (dots2) -- (vn);

    % Edge between vi and vip1 with labels above and below
    \draw[->] (vi) -- (vip1) node[midway, above] {$C(e_i) = w_i$} node[midway, below] {$P(e_i) = p_i$};

\end{tikzpicture}
    \caption{The line graph described Proposition~\ref{prop:nashcomp}, serving as physical graph for a contested route planning game for the \textsc{Knapsack} problem we reduce from. The interdiction costs $C$ for Red and penalties $P$ for Blue are set to match the \textsc{Knapsack} items' weights $w$ and values $p$.}
    \label{fig:knapsack}
\end{figure}
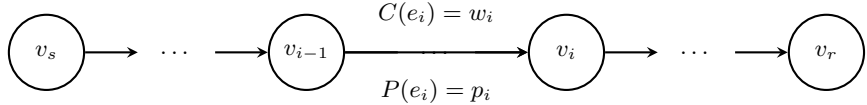

\begin{proposition}\label{prop:nashcomp}
    It is \NP-hard in terms of $|\pgraph|$ to find a Nash equilibrium for a contested route planning game with action spaces \ref{blue-flows} and \ref{red-interdiction}, and utility given in Formulation~\ref{utility}.
\end{proposition}
\begin{proof}
We prove \NP-hardness by reducing from the classical \textsc{Knapsack} problem~\cite{garey2002computers}. Consider an instance of \textsc{Knapsack} defined by a set of $n$ items indexed by $i \in [n]$, each with a nonnegative value $p_i$ and weight $w_i$, and a total weight budget $W$. The goal is to find a subset of items $S \subseteq [n]$ maximizing $\sum_{i \in S} p_i$ subject to the constraint $\sum_{i \in S} w_i \leq W$. We construct a contested route planning game on a line graph $\pgraph = (\pvertices, \pedges)$, where the nodes are $\pvertices = \{v_0, v_1, \ldots, v_n\}$ and edges $\pedges = \{e_i = (v_{i-1}, v_i) \mid i \in [n]\}$. The start and release points are $v_s = v_0$ and $v_r = v_n$ respectively. We define the interdiction cost function $C$ by $C(e_i) = w_i$, the interdiction penalty function $P$ by $P(e_i) = p_i$, and set the edge traversing penalty $T(e_i) = 1$ for all edges. The physical graph set like this is depicted in Figure~\ref{fig:knapsack}. Red’s interdiction budget is $B = W$. Since the graph contains a single path from $v_s$ to $v_r$, Blue’s action space consists only of that single path, making Blue’s choice trivial. Thus, Blue’s utility function reduces to $u(f,y) = \sum_{i=1}^n T(e_i) + \sum_{i=1}^n P(e_i) y_i = n + \sum_{i=1}^n p_i y_i$, where $y_i \in \{0,1\}$ denotes whether edge $e_i$ is interdicted. Red’s strategy is constrained by the budget $\sum_{i=1}^n w_i y_i \leq W$. Maximizing Blue’s penalty thus corresponds exactly to solving the \textsc{Knapsack} problem, since Red must select edges to maximize $\sum p_i y_i$ subject to the weight constraint. Therefore, finding a Nash equilibrium is \NP-hard.
\end{proof}
 % use a line graph, so Blue path is trivial, computing optimal red response is a knapsack problem so NP-h

Blue’s and Red’s (pure) best responses to fixed strategies $x_r$ and $x_b$ are
\begin{align*}
f^{\text{BR}} = \arg\min_{f \in \ref{blue-flows}} u(f, x_r)
\quad \text{and} \quad
y^{\text{BR}} = \arg\max_{y \in \ref{red-interdiction}} u(x_b, y),\quad\text{where}\\
% \end{equation*}
% where
% \begin{align*}
u(x_b, y) = \mathbb{E}_{f \sim x_b} [u(f, y)] \quad \text{and} \quad u(f, x_r) = \mathbb{E}_{y \sim x_r} [u(f, y)].
\end{align*}
While best response computation is closely related to finding a Nash equilibrium, these problems are generally distinct~\cite{xu2016mysteries}. However, in contested route planning games, computing Red’s best response is intractable. This follows immediately from the NP-hardness proof of finding a Nash equilibrium in Proposition~\ref{prop:nashcomp}, as it directly corresponds to solving the knapsack problem therein. Formally:

\begin{corollary}
    Let $\tilde{\ref{blue-flows}} \subseteq \ref{blue-flows}$ be of size $k$ (possibly smaller than $|\ref{blue-flows}|$) and $\widetilde{x}_b$ be a distribution with support $\tilde{\ref{blue-flows}}$.
    Finding Red's best response against $\widetilde{x}_b$ in a contested route planning game with utility~\ref{utility} is \NP-hard in terms of $|\pgraph|$ and $k$.
\end{corollary}

Nevertheless, using the action space formulation~\ref{red-interdiction}, Red’s best response can be expressed as the following polynomial-sized MILP:
\begin{equation}
\tag{Red-BR}
    % \begin{aligned}
        \max_{y\in\ref{red-interdiction}} \sum_{f\in\tilde{\ref{blue-flows}}}\sum_{e \in \pedges} \widetilde{x}_b(f)\left(T(e) f(e) + P(e) f(e) y(e)\right),
    % \end{aligned}
\end{equation}
which can be solved using modern MILP solvers like Gurobi or CPLEX. However, even when Blue’s strategy $\widetilde{x}_b$ is randomized, computing the best response remains a \textsc{Knapsack} instance in disguise. Each edge $e \in \pedges$ can be thought of as a single item; the interdiction cost $C(e)$ becomes the item's weight, and the effective item value $p(e)$ is given by
\begin{equation*}
    p(e) = \sum_{f\in\tilde{\ref{blue-flows}}} \widetilde{x}_b(f)\left(T(e) f(e) + P(e) f(e)\right).
\end{equation*}
The objective is then to choose a subset of items (edges) whose total weight does not exceed Red’s budget $B$, and whose total value is maximized -- precisely as in the standard \textsc{Knapsack} structure. This connection allows us to leverage specialized algorithms for the \textsc{Knapsack} problem~\cite{pisinger1997minimal, pisinger1998fast, martello1990knapsack}, including highly optimized implementations such as KPYM~\cite{kpym}, which can outperform general-purpose MILP solvers. Notably, if either the interdiction costs $C$ or interdiction penalties $P$ are integral, exact solutions can be found in pseudo-polynomial time using dynamic programming~\cite{martello1990knapsack}. Furthermore, relaxing Red’s action space to allow partial interdiction of edges renders the problem immediately tractable in polynomial time~\cite{martello1990knapsack}. In contrast, computing Blue’s best response is tractable right from the start and does not pose the same level of complexity.

\begin{proposition}
    Let $\tilde{\ref{red-interdiction}} \subseteq \ref{red-interdiction}$ be of size $k$ (possibly smaller than $|\ref{red-interdiction}|$) and $\widetilde{x}_r$ be a distribution with support $\tilde{\ref{red-interdiction}}$.
    Finding Blue's best response against $\widetilde{x}_r$ in a contested route planning game with utility~\ref{utility} is polynomial in terms of $|\pgraph|$ and $k$.
\end{proposition}
\begin{proof}
Using the action space formulation \ref{blue-flows}, the problem of computing Blue's best response reduces to solving the following polynomial-sized MILP:
\begin{equation}
\tag{Blue-BR}
    % \begin{aligned}
        \min_{f\in\ref{blue-flows}} \sum_{y\in\tilde{\ref{red-interdiction}}}\sum_{e \in \pedges} \widetilde{x}_r(y)\left(T(e) f(e) + P(e) f(e) y(e)\right),
    % \end{aligned}
\end{equation}
where the decision variable $f$ corresponds to a unit $v_s$-$v_r$ flow over the graph $\pgraph$. The constraint matrix of this flow polytope is totally unimodular, and the right-hand side is integral; thus, the feasible region forms an integral polyhedron. 
Consequently, solving the LP relaxation yields a solution to MILP\cite{schrijver1998theory}.
\end{proof}

Furthermore, observe that the objective function is linear in $f$ and decomposes additively across edges. For a fixed distribution $\widetilde{x}_r$, we may define a deterministic edge length function:
\begin{equation*}
    \ell(e) = \sum_{y \in \tilde{\ref{red-interdiction}}} \widetilde{x}_r(y) \left( T(e) + P(e) y(e) \right).
\end{equation*}
Under this transformation, the best response for Blue reduces to finding a path from $v_s$ to $v_r$ of minimum total length with respect to the modified edge lengths $\ell(e)$. Since all edge lengths remain non-negative, standard shortest path algorithms such as Dijkstra’s algorithm apply directly. Moreover, if the original penalty function $T$ supports an admissible heuristic (e.g., a Euclidean straight-line distance), this admissibility is preserved in $\ell$, enabling the use of informed search techniques such as $A^*$~\cite{hart1968formal}.

Finally, note that this polynomial solvability is robust to a range of possible extensions. Even if additional movement constraints are imposed on Blue’s path, for example, forbidden edge sequences, vehicle dynamics, or timing restrictions, as long as the feasible set remains representable as a network flow with a totally unimodular constraint matrix, the best response remains polynomial-time solvable in $k$ and $|\pgraph|$.

\subsection{Approximating the Equilibrium Using Strategy Generation}

Although the strategy spaces of both players are exponential in size, route planning and patrolling games on real-world graphs often exhibit Nash equilibria supported on relatively few pure strategies~\cite{cerny2024layered,vcerny2024contested,krever2025guard}. This observation motivates the use of the double oracle (DO) algorithm, a widely used iterative method for solving large-scale zero-sum games with efficient best-response oracles.

The DO algorithm incrementally builds a restricted subgame, defined by subsets of pure strategies for each player, with the aim of converging to an equilibrium of the full game. In the context of contested route planning, these pure strategies correspond to Blue’s feasible routing plans $f \in \ref{blue-flows}$ and Red’s interdiction plans $y \in \ref{red-interdiction}$. At each iteration, the algorithm maintains subsets $\widetilde{\ref{blue-flows}} \subseteq \ref{blue-flows}$ and $\widetilde{\ref{red-interdiction}} \subseteq \ref{red-interdiction}$ representing the current subgame.

The algorithm begins with small initial subsets of strategies for each player. At every iteration, it computes a Nash equilibrium $(\widetilde{x}_b^*, \widetilde{x}_r^*)$ of the subgame using a standard linear program over the restricted support of actions in the subgame. Each player then computes a best response against the opponent’s mixed equilibrial strategy in the subgame using the respective best-response oracles, solving \textsc{Blue-BR} or \textsc{Red-BR} as needed. The best responses can be implemented in various ways, as we explained in the previous section, using different routing and knapsack solving algorithms, for instance. If these best responses are already present in the subgame, the algorithm terminates, returning the equilibrium of the subgame as an equilibrium of the full game. Otherwise, the new strategies are added to the respective subsets, and the process repeats.

Exact termination occurs when neither player can improve their utility by deviating from the current subgame support. However, in practice, it is often sufficient to terminate when the equilibrium gap, defined as $\nabla = u(\widetilde{x}_b^*, y^\text{BR}) - u(f^\text{BR}, \widetilde{x}_r^*)$, falls below a predefined tolerance $\epsilon > 0$. In that case, the strategy pair $(\widetilde{x}_b^*, \widetilde{x}_r^*)$ forms a $2\epsilon$-approximate Nash equilibrium.

\begin{algorithm}[t]
\caption{Double Oracle for Contested Route Planning Games}\label{alg:do}
\begin{algorithmic}[1]
% \REQUIRE $\mathcal{G}_d, \mathcal{G}_a, \targetval^\odot, \edgeadj, u, \epsilon > 0$
% \ENSURE $y = x^n$
\STATE $\widetilde{F}, \widetilde{Y} \gets \textsc{InitialSubgame}(F, Y)$
% \STATE $\nabla \gets \infty$
\REPEAT
\STATE $\widetilde{x}_b^*, \widetilde{x}_r^* \gets \textsc{NashEquilibrium}(\widetilde{F}, \widetilde{Y})$
\STATE $f^\text{BR}, y^\text{BR} \gets \textsc{BlueBR}(\widetilde{x}_r^*), \textsc{RedBR}(\widetilde{x}_b^*)$ 
\STATE $\widetilde{F}, \widetilde{Y} \gets \widetilde{F}\cup \{f^\text{BR}\}, \widetilde{Y}\cup \{y^\text{BR}\}$
\UNTIL{$\textsc{EquilibriumGap}(\widetilde{x}_b^*, \widetilde{x}_r^*, f^\text{BR}, y^\text{BR}) \leq \epsilon$}
\end{algorithmic}
\end{algorithm}
\section{Empirical Evaluation}

We now turn to experiments on contested route planning scenarios. The experiments use real-world urban network data from two regions: Southern California between Los Angeles and San Diego, and the Donetsk Oblast in eastern Ukraine, a current frontline in the ongoing Russo-Ukrainian conflict. Based on these maps, three routing scenarios were constructed.

The first two scenarios are set in Southern California and are part of a broader logistics model provided to us by the U.S. Office of Naval Research. This model has been calibrated to closely resemble training scenarios from the Marine Corps Air-Ground Combat Center at Twentynine Palms, used to train and certify Marine Littoral Regiments. The logistics network consists of two supply locations in Barstow and Twentynine Palms, CA, and 13 expeditionary advanced bases (EABs) situated along or near the Pacific coast. For our purposes, the first scenario considers routing between EABs located in MCAS Miramar and Carlsbad, while the second scenario involves routing between the Twentynine Palms supply point and the Carlsbad EAB.

The third scenario is situated in Ukraine during the ongoing Russo-Ukrainian conflict. As of mid-2025, the area around Toretsk and Niu-York remains a contested frontline, with continued fighting including artillery exchanges and Russian efforts to advance toward key Ukrainian defensive positions. Ukrainian forces have reinforced the area near Toretsk to counter further Russian movements toward Kostiantynivka and the broader Donetsk front. In Niu-York, only a few kilometers to the east, persistent skirmishes and drone strikes have posed threats to both civilian and military infrastructure. This third scenario is designed to mimic supply runs from Kramatorsk to Niu-York.

The routing data are sourced from OpenStreetMap (OSM). The Red and Blue models, although simplified for the purposes of this work, were developed in cooperation with the U.S. Office of Naval Research to reflect an adequate level of operational detail. The penalties~$P$ were assigned as follows: regular edges have $P(e) = 1$ (a ``low probability-of-kill (p-k)'' edge penalty in military terms), while edges tagged as ``bridge'' (which may also include certain elevated structures) in OSM have $P(e) = 3$ (a ``high p-k'' edge penalty). This binary penalty structure allows the computation of path throughput against a Red strategy~$x_r$ as
    $\mathbb{E}_{y \sim x_r}\prod_{e\in y}p(e)$, 
where $p(e)$ denotes the edge’s p-k, set to 0.5 for high-p-k edges and 0.2 for low-p-k edges. For edge interdiction costs, the costs are set proportional to the distances from the start and release points, according to
    $C(e) = \frac{1.0}{\min\left(|e,v_s|, |e,v_r|\right)}$,
and then normalized as
    $
    C(e) = \left\lceil 0.8 + \frac{(C(e) - \min_{e'} C(e')) \cdot (B + 1.3)}{\max_{e'} C(e') - \min_{e'} C(e')} \right\rceil.
    $
In all scenarios, we solve games for Red budgets $B$ ranging from 1 to 6.

All experiments were performed on an Intel Xeon Gold 6226 (2.9\,GHz) with 32\,GB of RAM. Linear programs were solved using the HiGHS Optimizer 1.9~\cite{huangfu2018parallelizing}, on a 64-bit Linux platform. Map data were obtained using osmnx 2.0.1, and the double oracle algorithm was implemented in Python 3.9.21 with an equilibrium gap of $\epsilon=0.1$. We used our own, non-optimized implementations of the A$^*$ algorithm as the Blue player oracle and a dynamic programming approach for the knapsack problem as the Red player oracle. Runtime results and the number of double oracle iterations required for convergence are summarized in Table~\ref{tab:scal}.

\begin{table}[h!]
\centering
\caption{Runtimes [s] and number of iterations of the double oracle.}
\begin{tabular}{l||cc|cc|cc|cc|cc|cc}
 & \multicolumn{2}{c|}{budget 1} 
 & \multicolumn{2}{c|}{budget 2} 
 & \multicolumn{2}{c|}{budget 3} 
 & \multicolumn{2}{c|}{budget 4} 
 & \multicolumn{2}{c|}{budget 5} 
 & \multicolumn{2}{c}{budget 6} \\
 & time & iters
 & time & iters
 & time & iters
 & time & iters
 & time & iters
 & time & iters \\
\hline
SoCal 1 & 2.44 & 37 & 3.41 & 50 & 2.96 & 41 & 6.36 & 77 & 6.82 & 81 & 8.82 & 99 \\
SoCal 2 & 20.07 & 53 & 24.45 & 59 & 44.26 & 88 & 43.66 & 85 & 71.51 & 127 & 84.08 & 140 \\
Ukraine & 1.24 & 19 & 1.35 & 23 & 1.75 & 32 & 1.99 & 36 & 2.20 & 40 & 2.67 & 49
\end{tabular}
\label{tab:scal}
\end{table}

The remainder of this section is devoted to a qualitative analysis. We examine how the optimal game-theoretic strategies change as the Red budget increases. These strategies are compared in terms of expected throughput to two baselines: the fastest route between the start and release points, and a deterministic Red-aware solution designed to reduce the number of high probability-of-kill attacks. Finally, we assess the robustness of the computed game-theoretic solutions against Red players with either increased or decreased budgets.

\subsection{Scenario 1: Southern California -- MCAS Miramar to Carlsbad }

\begin{figure}[t]
    \centering
    % First image
    \begin{minipage}[b]{0.49\textwidth}
        \centering
        \includegraphics[draft=false,width=\linewidth]{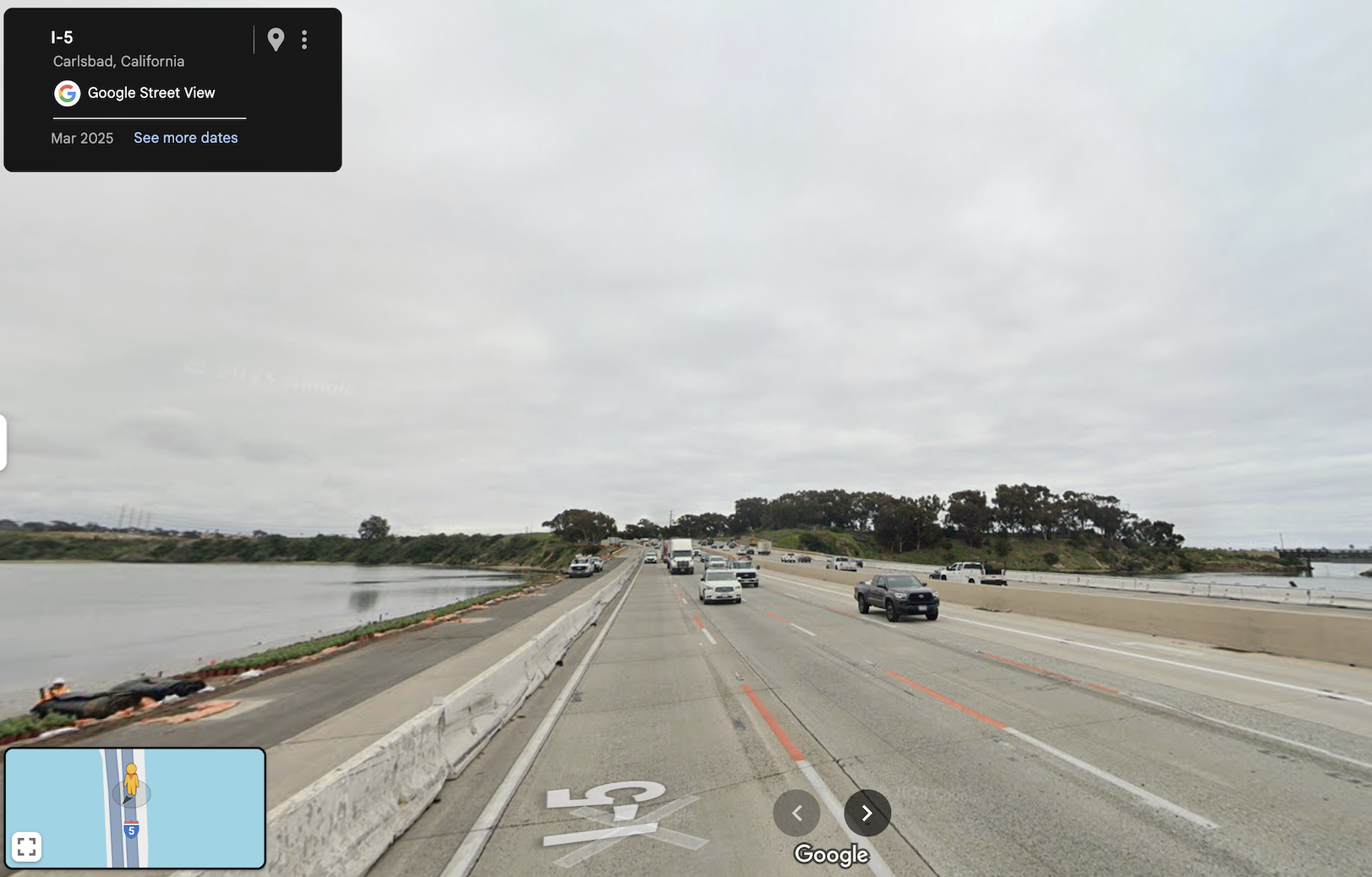}
    \end{minipage}
    \hfill
    % Second image
    \begin{minipage}[b]{0.24\textwidth}
        \centering
        \includegraphics[width=\linewidth]{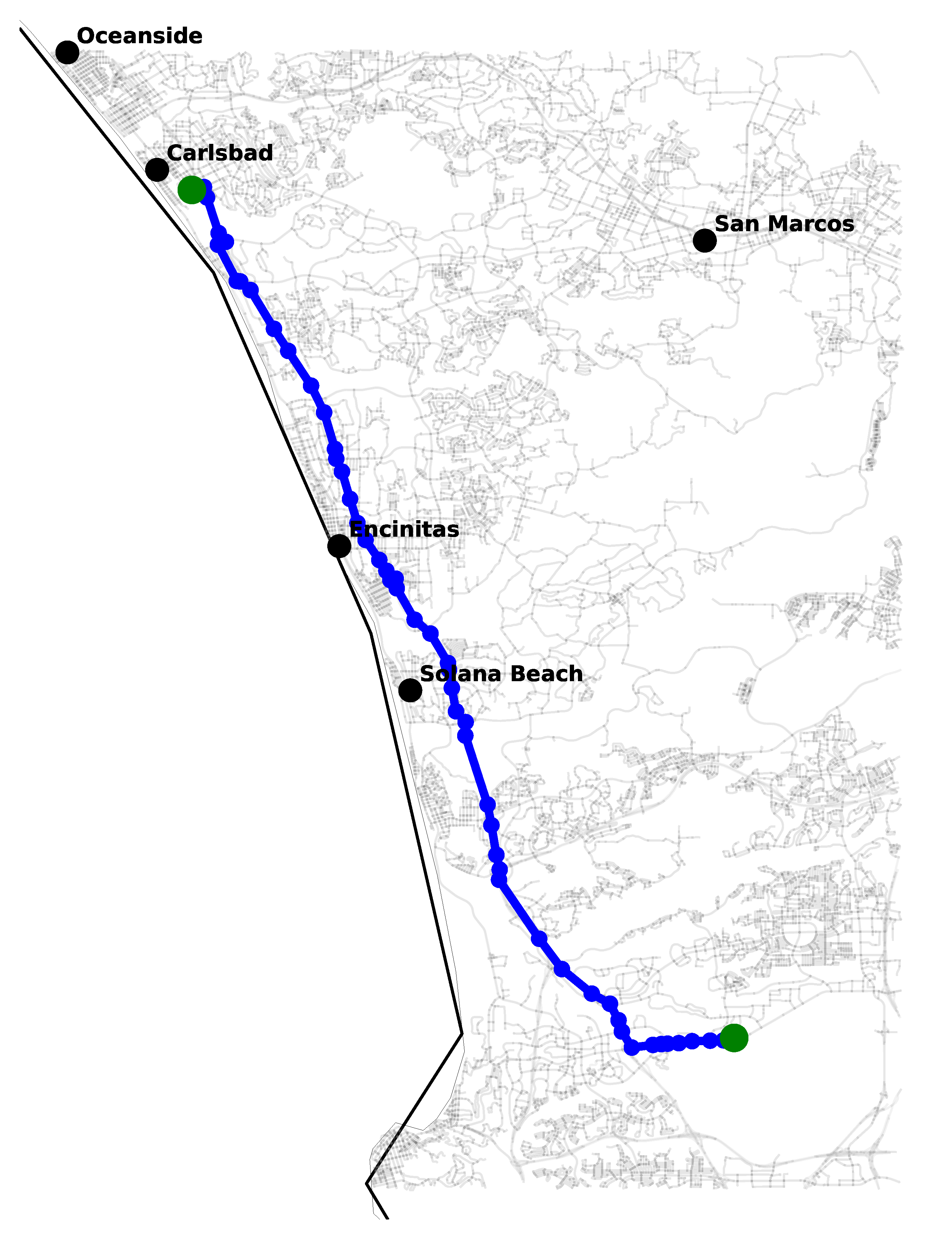}
    \end{minipage}
    \hfill
    % Third image
    \begin{minipage}[b]{0.24\textwidth}
        \centering
        \includegraphics[width=\linewidth]{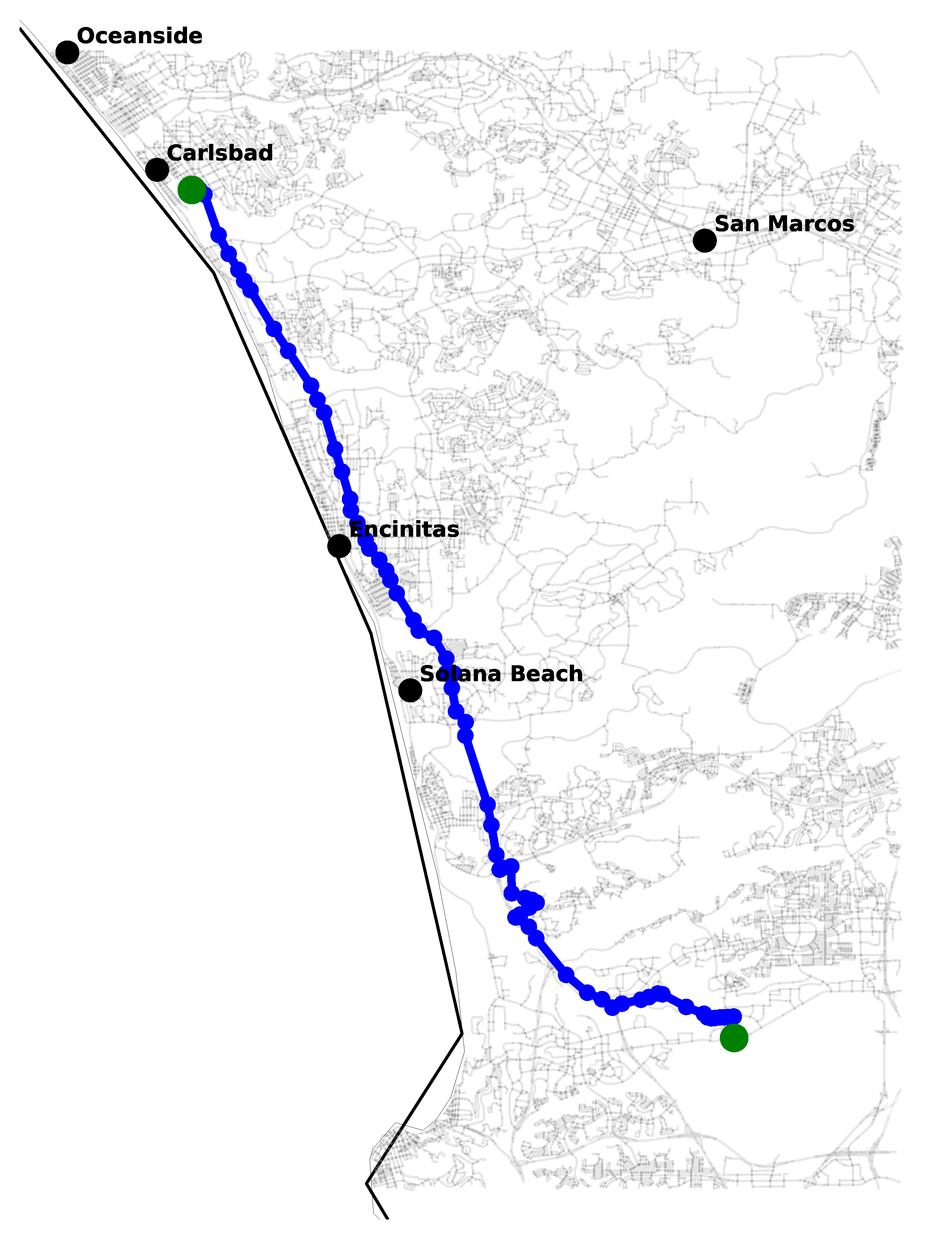}
    \end{minipage}

    \caption{Typical scenery and fastest and red-aware routes in Scenario 1.}
    \label{fig:sc1:fp}
\end{figure}

The first scenario examines route planning between two locations connected by a major highway as well as several smaller roads. Both the start and release locations are situated in urban areas surrounded by local street networks. The central highway features many elevated segments, which constitute high p-k edges, as illustrated on the left in Figure~\ref{fig:sc1:fp}. The graph obtained from OpenStreetMap includes 10,440 nodes and 28,316 edges, representing a realistic problem size.

The game-theoretic solutions for this scenario are shown in Figure~\ref{fig:sc1-gt}. The probabilities of the individual paths chosen by Blue are omitted. Each edge in Red’s support is highlighted in red, with greater hue intensity indicating higher aggregated probability. Notably, this scenario does not exhibit a significant increase in the size of the Blue player’s support as Red’s budget increases. We attribute this to the large number of available high p-k targets. Across all cases, the throughput remains above 70\%. Table~\ref{tab:sc1-tablh} reports the probabilities of encountering different numbers of low and high p-k attacks under the various game-theoretic solutions. It is worth noting that across all Red budgets, there is a relatively high probability that Blue would avoid any attack entirely.
\begin{table}[h!]
\centering
\caption{Probabilities of low (L) and high (H) p-k encounters in Scenario 1.}
\begin{tabular}{c||c|c|c|c|c|c}
~~Red budget~~ & ~~no attack~~ & ~~~1L~~~ & ~~~1H~~~ & ~1L/1H~ & ~~~2H~~~ & ~~~3H~~~ \\
\hline
1 & 0.83 & - & 0.17 & - & - & - \\
2 & 0.67 & - & 0.33 & - & <0.01 & - \\
3 & 0.5 & - & 0.49 & - & <0.01 & <0.01 \\
4 & 0.53 & <0.01 & 0.39 & <0.01 & 0.08 & <0.01 \\
5 & 0.29 & <0.01 & 0.71 & <0.01 & <0.01 & <0.01 \\
6 & 0.14 & - & 0.86 & - & <0.01 & -
\end{tabular}
\label{tab:sc1-tablh}
\end{table}

\begin{figure}[t]
    \centering

    % First row
    \begin{subfigure}[b]{0.32\textwidth}
        \includegraphics[width=\linewidth]{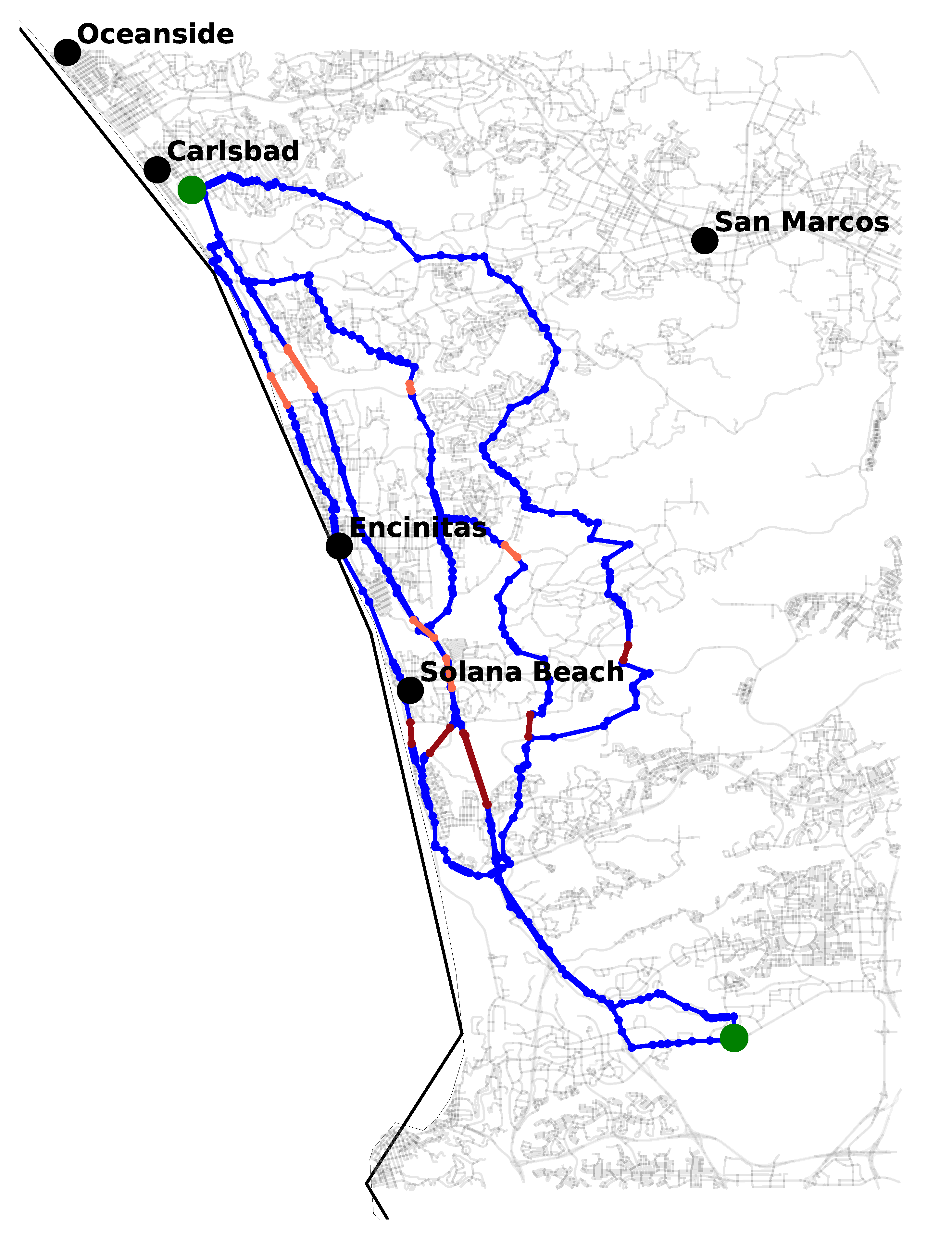}
        \caption{Red budget 1}
    \end{subfigure}
    \hfill
    \begin{subfigure}[b]{0.32\textwidth}
        \includegraphics[width=\linewidth]{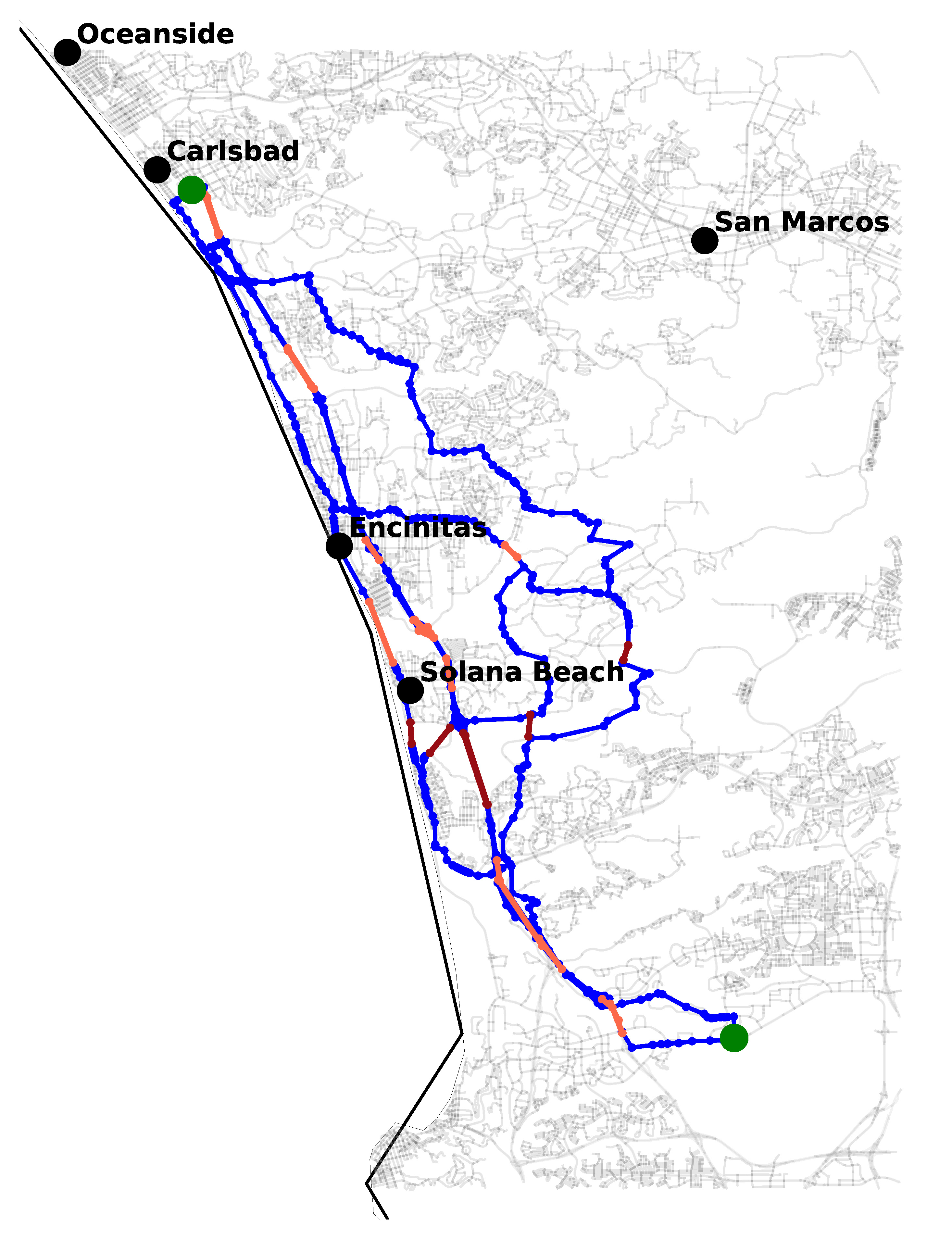}
        \caption{Red budget 2}
    \end{subfigure}
    \hfill
    \begin{subfigure}[b]{0.32\textwidth}
        \includegraphics[width=\linewidth]{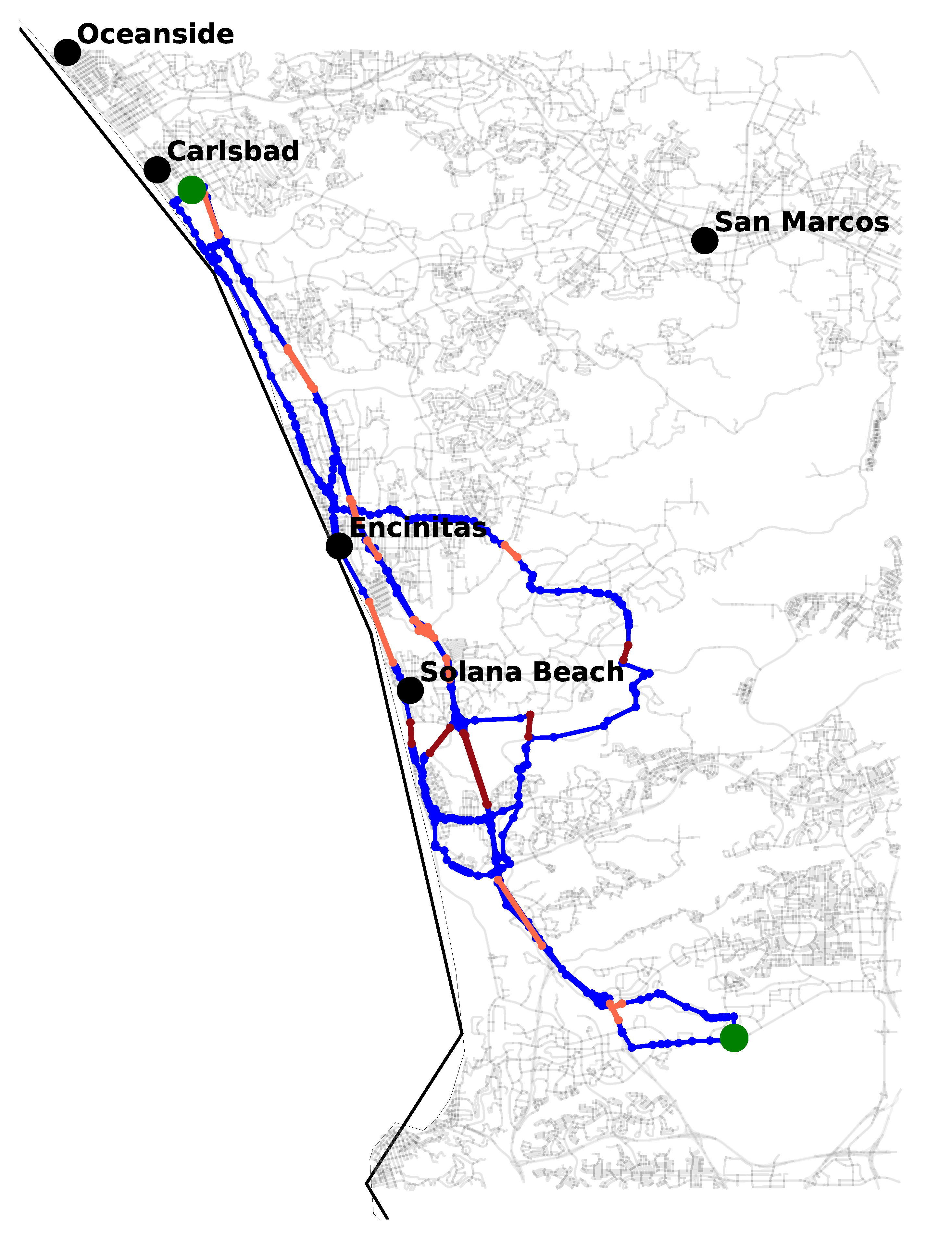}
        \caption{Red budget 3}
    \end{subfigure}

    % Second row
    \begin{subfigure}[b]{0.32\textwidth}
        \includegraphics[width=\linewidth]{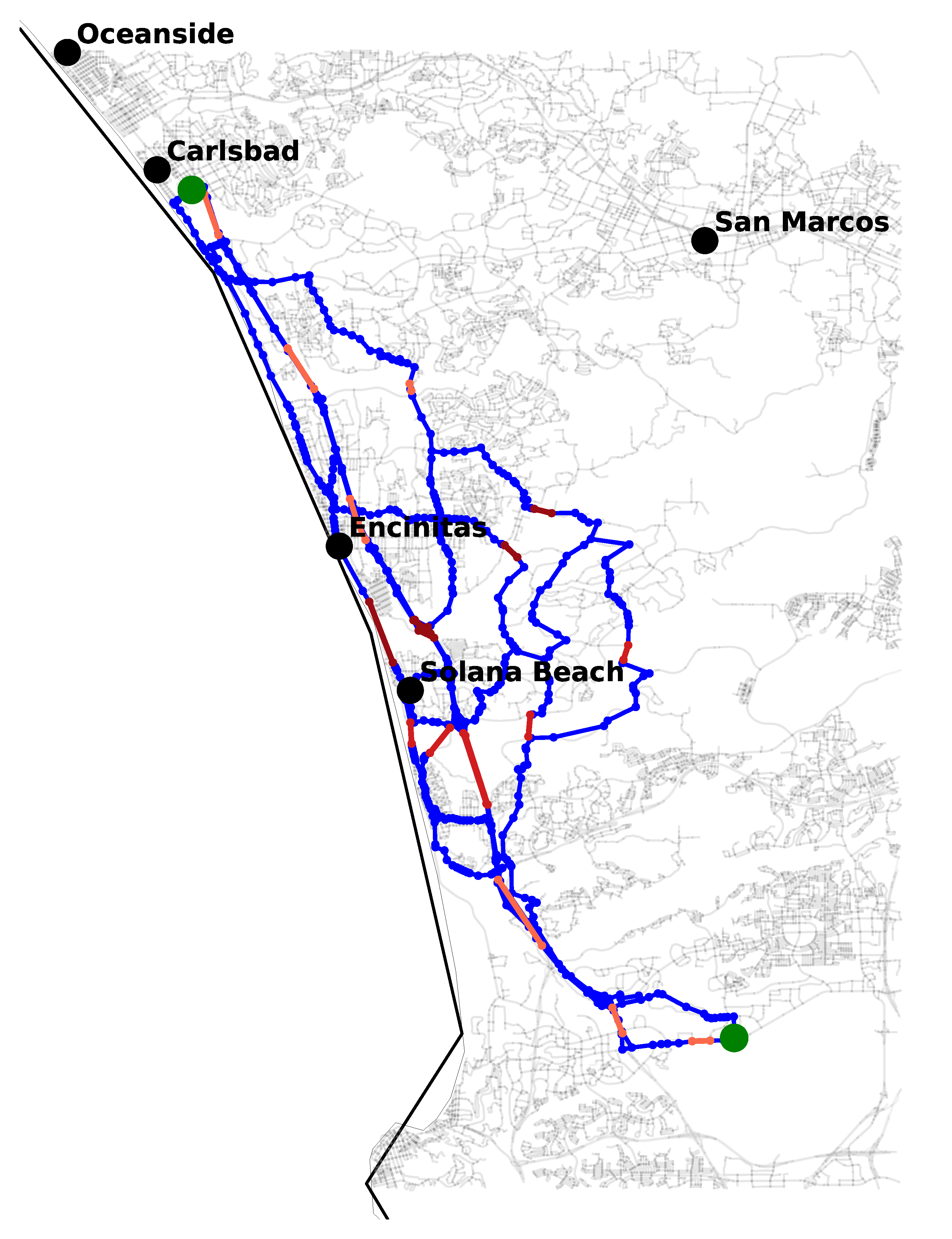}
        \caption{Red budget 4}
    \end{subfigure}
    \hfill
    \begin{subfigure}[b]{0.32\textwidth}
        \includegraphics[width=\linewidth]{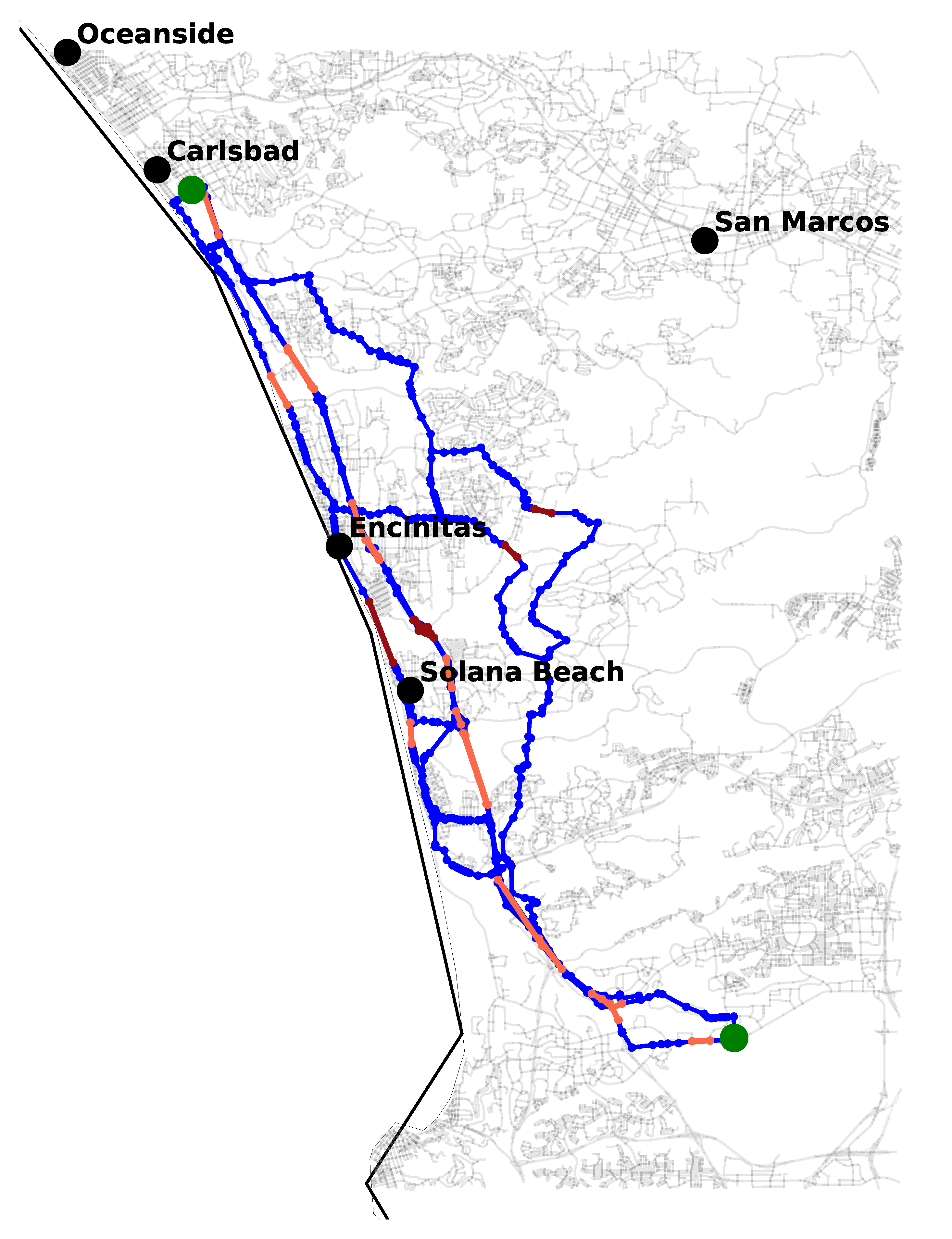}
        \caption{Red budget 5}
    \end{subfigure}
    \hfill
    \begin{subfigure}[b]{0.32\textwidth}
        \includegraphics[width=\linewidth]{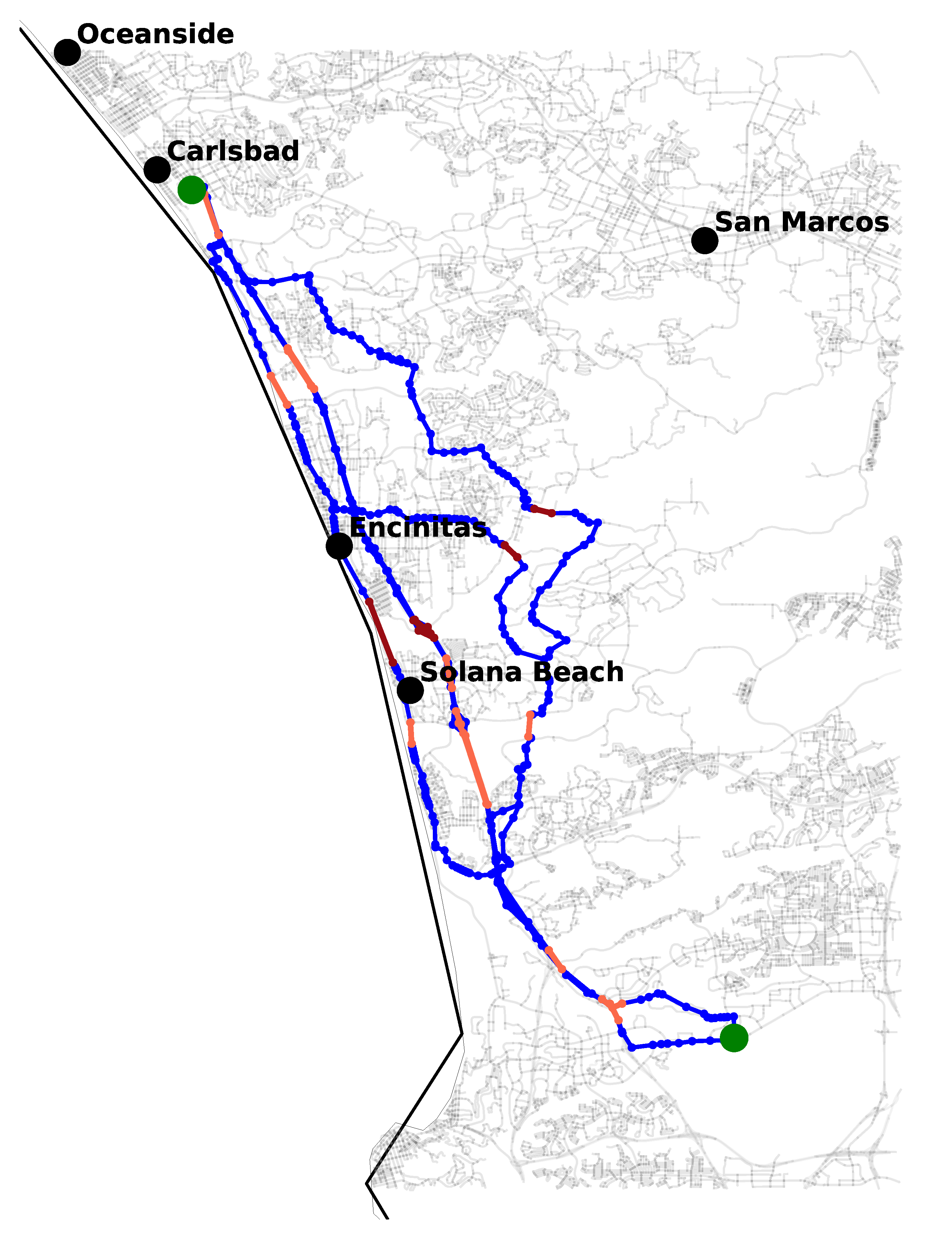}
        \caption{Red budget 6}
    \end{subfigure}

    \caption{Computed game-theoretic solutions in Scenario 1.}
    \label{fig:sc1-gt}
\end{figure}

\begin{table}[h!]
\centering
\caption{Throughput of solutions against different Red models in Scenario 1.}
\begin{tabular}{c||c|c|c|c|c|c}
Exp.$\backslash$True budget & ~~~~1~~~~ & ~~~~2~~~~ & ~~~~3~~~~ & ~~~~4~~~~ & ~~~~5~~~~ & ~~~~6~~~~ \\
\hline
1 & 0.90 & 0.50 & 0.55 & 0.45 & 0.35 & 0.38 \\
2 & 0.83 & 0.83 & 0.79 & 0.71 & 0.69 & 0.72 \\
3 & 0.83 & 0.75 & 0.85 & 0.77 & 0.73 & 0.72 \\
4 & 0.86 & 0.71 & 0.64 & 0.75 & 0.73 & 0.68 \\
5 & 0.86 & 0.71 & 0.57 & 0.63 & 0.73 & 0.72  \\
6 & 0.86 & 0.79 & 0.57 & 0.61 & 0.62 & 0.70
\end{tabular}
\label{tab:sc1-robust}
\end{table}

For comparison, the fastest route and a Red-aware deterministic path are depicted on the right in Figure~\ref{fig:sc1:fp}. Among the 54 edges of the fastest route, 15 are classified as high p-k. Given the edge costs, with a budget of 6, Red’s optimal strategy is to deploy 5 high p-k attacks, reducing throughput to 3.12\%. If the Red-aware path is selected instead, Red can subject it to 3 high p-k and 1 low p-k attack within the same budget, reducing throughput to 10\%. In other words, the game-theoretic solution improves throughput by more than 22 and 7 times compared to these deterministic alternatives.

Finally, Table~\ref{tab:sc1-robust} presents a robustness analysis of the game-theoretic solutions, illustrating how throughput changes if the computed solution is challenged by Red with a different budget. Apart from the smallest budget case, where for higher budgets the throughput can fall to almost a half of the optimized value, the solutions remain relatively robust, due to the network’s density and the large number of high p-k edges.

\subsection{Scenario 2: Southern California --  Twentynine Palms to Carlsbad}

The second scenario is significantly larger than the first, with the OpenStreetMap-derived graph containing 56,410 nodes and 158,676 edges, although the network itself is considerably sparser. This area includes numerous smaller towns and villages connected by highways and expressways that are frequently exposed, creating bottlenecks (visible in the solutions) that are highly susceptible to Red’s attacks. An example of the typical terrain in this region is shown on the left in Figure~\ref{fig:sc2-fp}. The game-theoretic solutions for this scenario are presented in Figure~\ref{fig:sc2-gt}. Possibly due to the lower proportion of high-p-k edges relative to the overall graph size, the solutions display a pattern also observed in the third scenario: the complexity of Blue’s strategies increases as Red’s budget grows. Table~\ref{tab:sc2-pk} reports the expected p-k attacks. It is notable that, due to the presence of bottlenecks, certain types of attacks become unavoidable for Red budgets between 2 and 4, while lower and higher budgets offer more flexibility to Blue.

\begin{figure}[t]
    \centering

    % First image
    \begin{minipage}[b]{0.56\textwidth}
        \centering
        \includegraphics[draft=false,width=\linewidth]{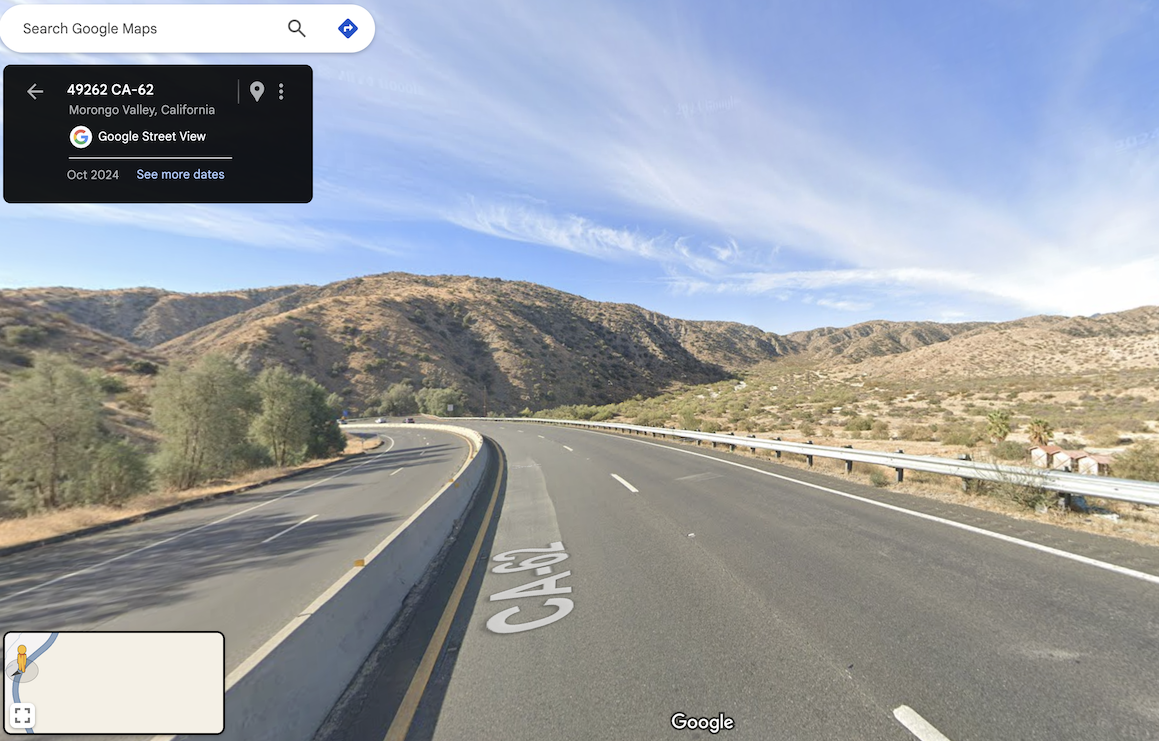}
        % \caption*{(a) Description 1}
    \end{minipage}
    \hfill
    % Second image
    \begin{minipage}[b]{0.38\textwidth}
        \centering
        \includegraphics[width=\linewidth]{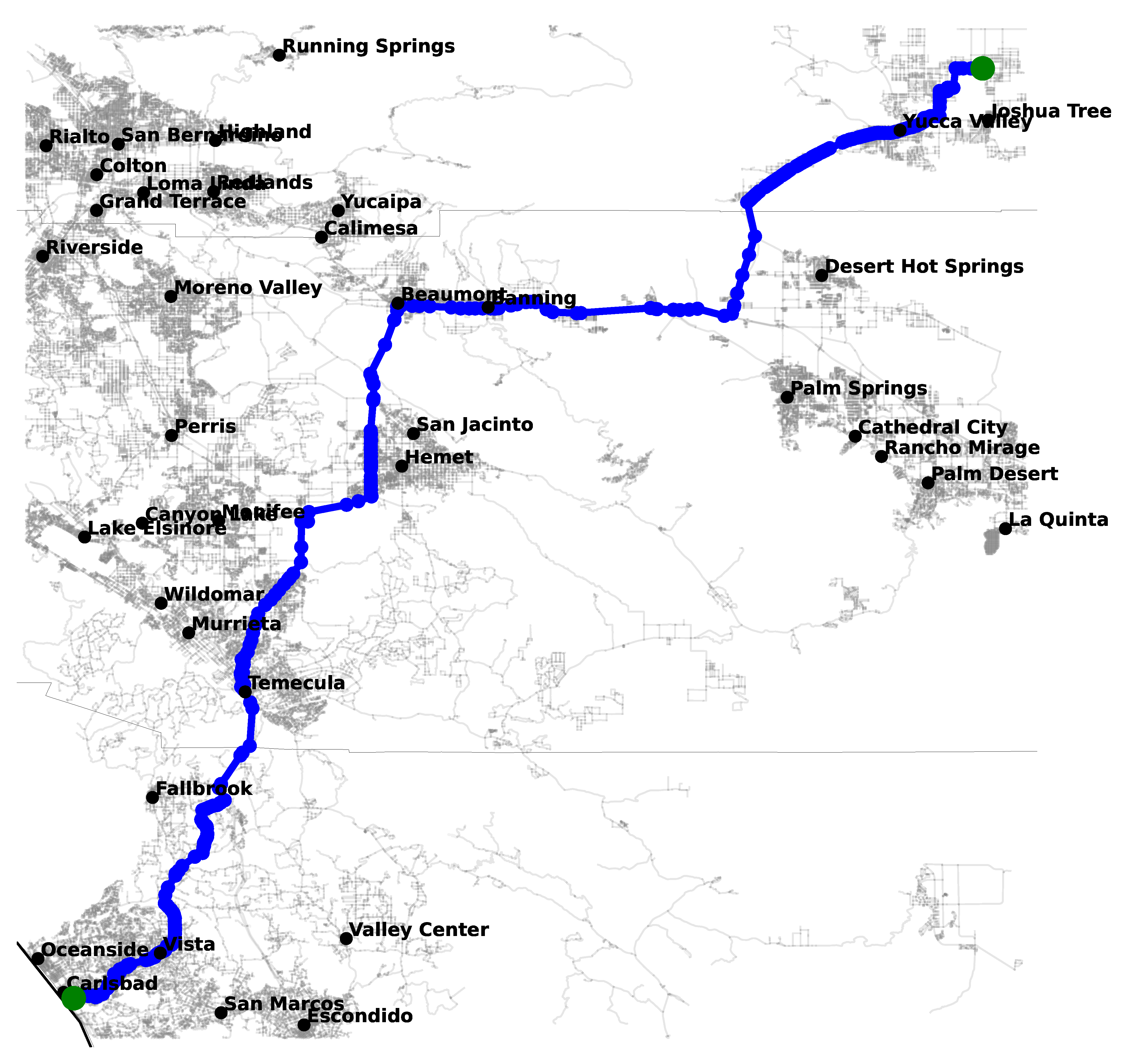}
        % \caption*{(b) Description 2}
    \end{minipage}

    \caption{Typical scenery and the fastest route in Scenario 2.}
    \label{fig:sc2-fp}
\end{figure}

\begin{figure}[p]
    \centering

    % Row 1
    \begin{subfigure}[b]{0.48\textwidth}
        \includegraphics[width=\linewidth]{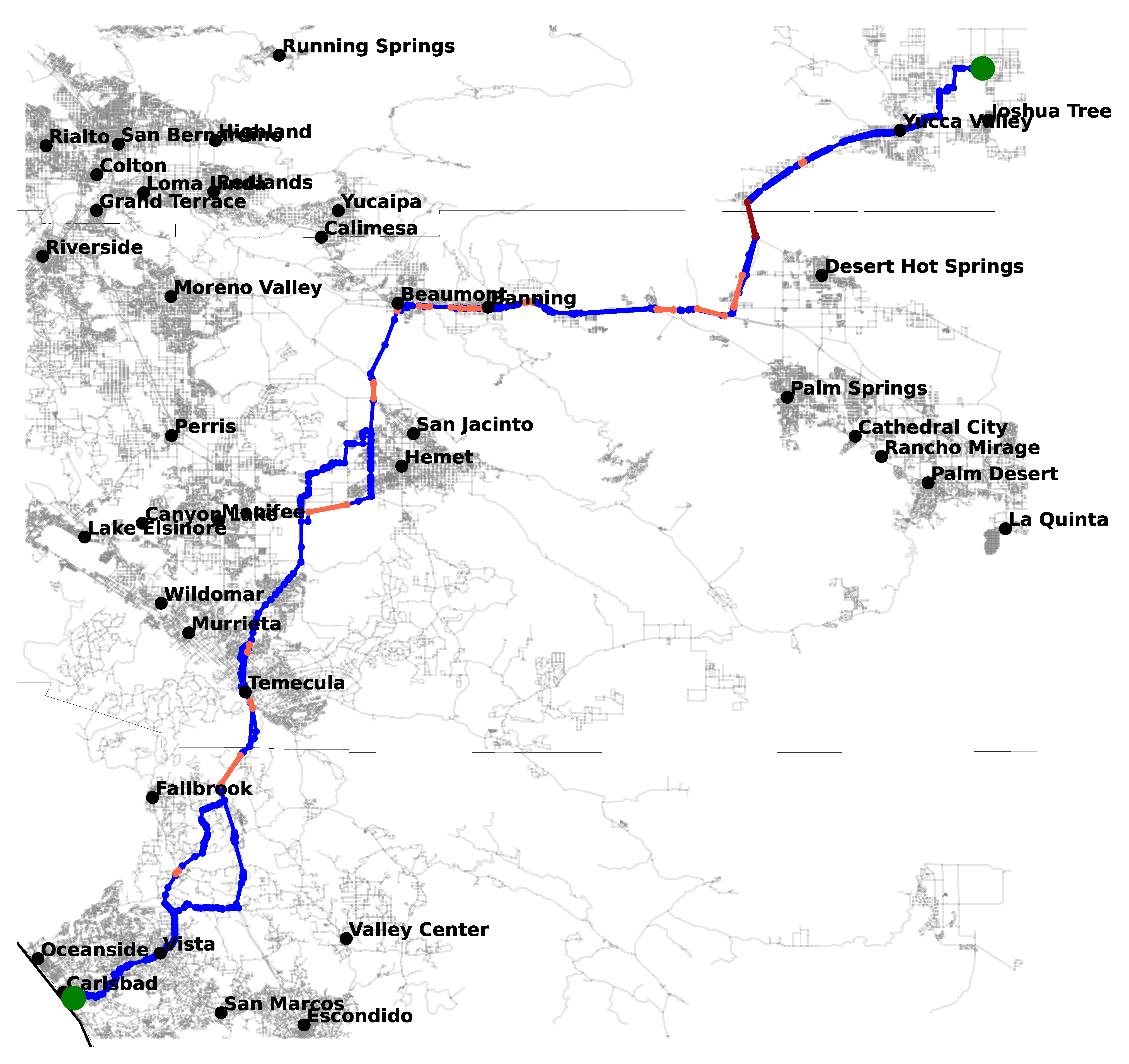}
        \caption{Red budget 1}
    \end{subfigure}
    \hfill
    \begin{subfigure}[b]{0.48\textwidth}
        \includegraphics[width=\linewidth]{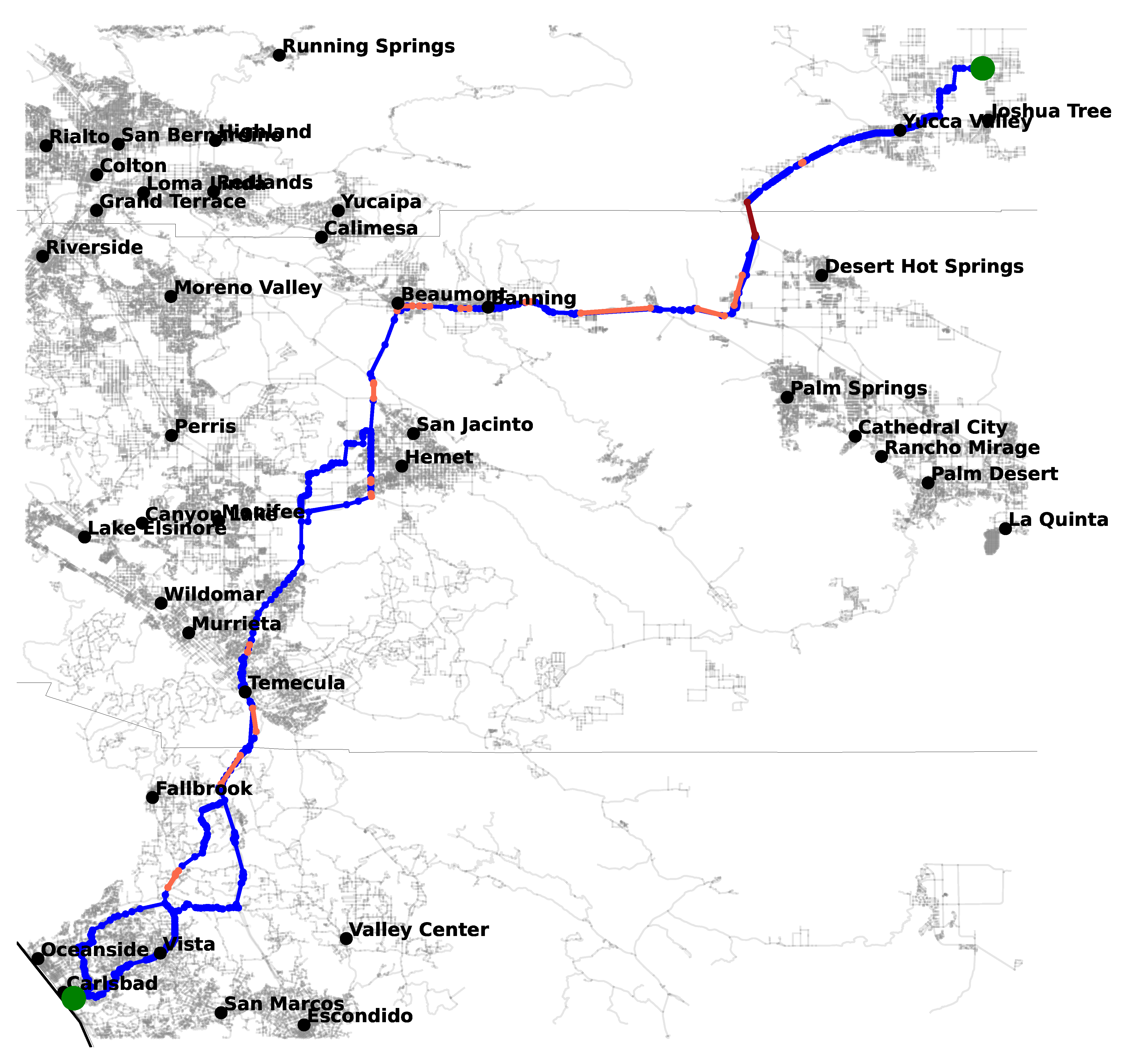}
        \caption{Red budget 2}
    \end{subfigure}

    % Row 2
    \begin{subfigure}[b]{0.48\textwidth}
        \includegraphics[width=\linewidth]{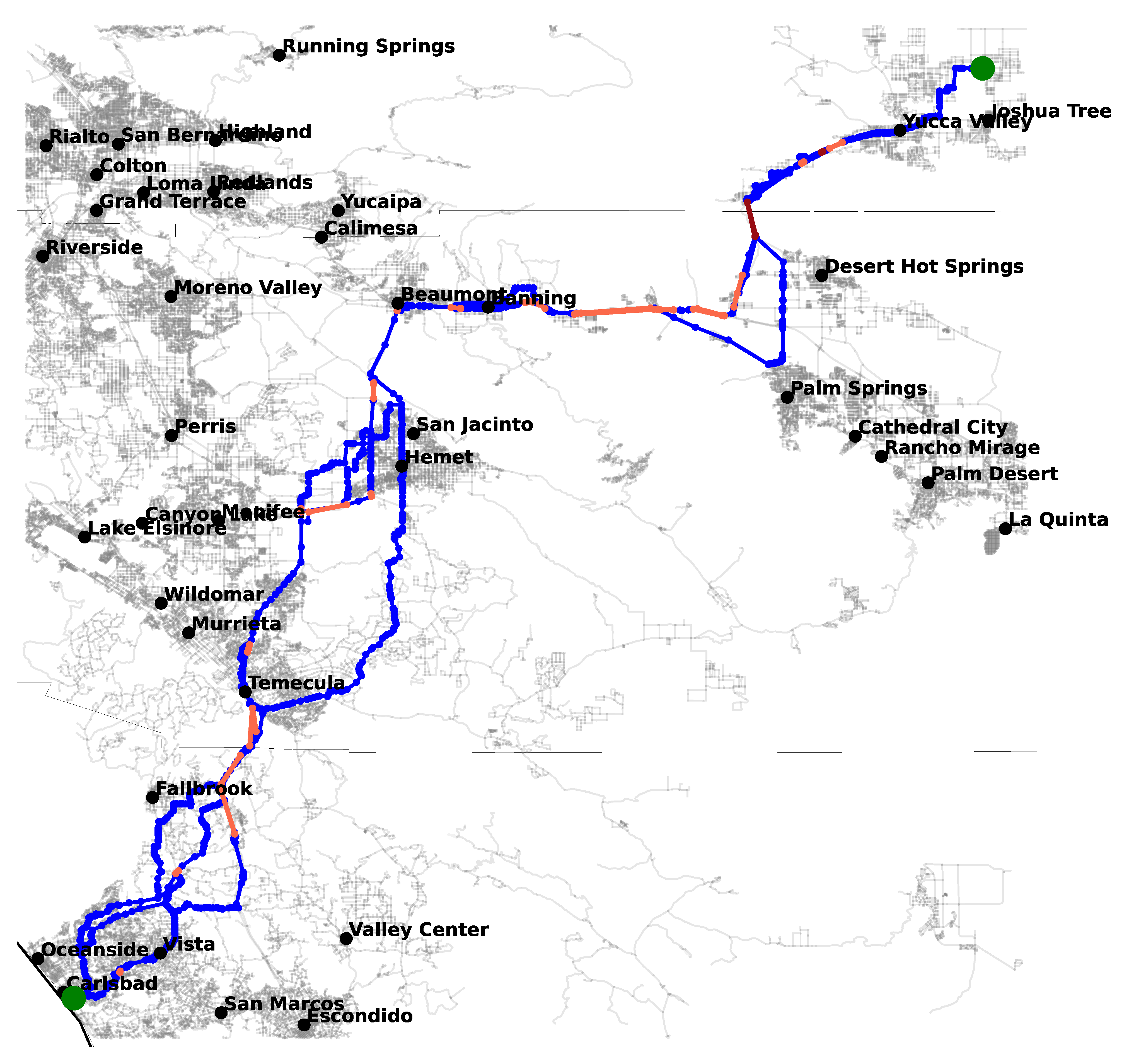}
        \caption{Red budget 3}
    \end{subfigure}
    \hfill
    \begin{subfigure}[b]{0.48\textwidth}
        \includegraphics[width=\linewidth]{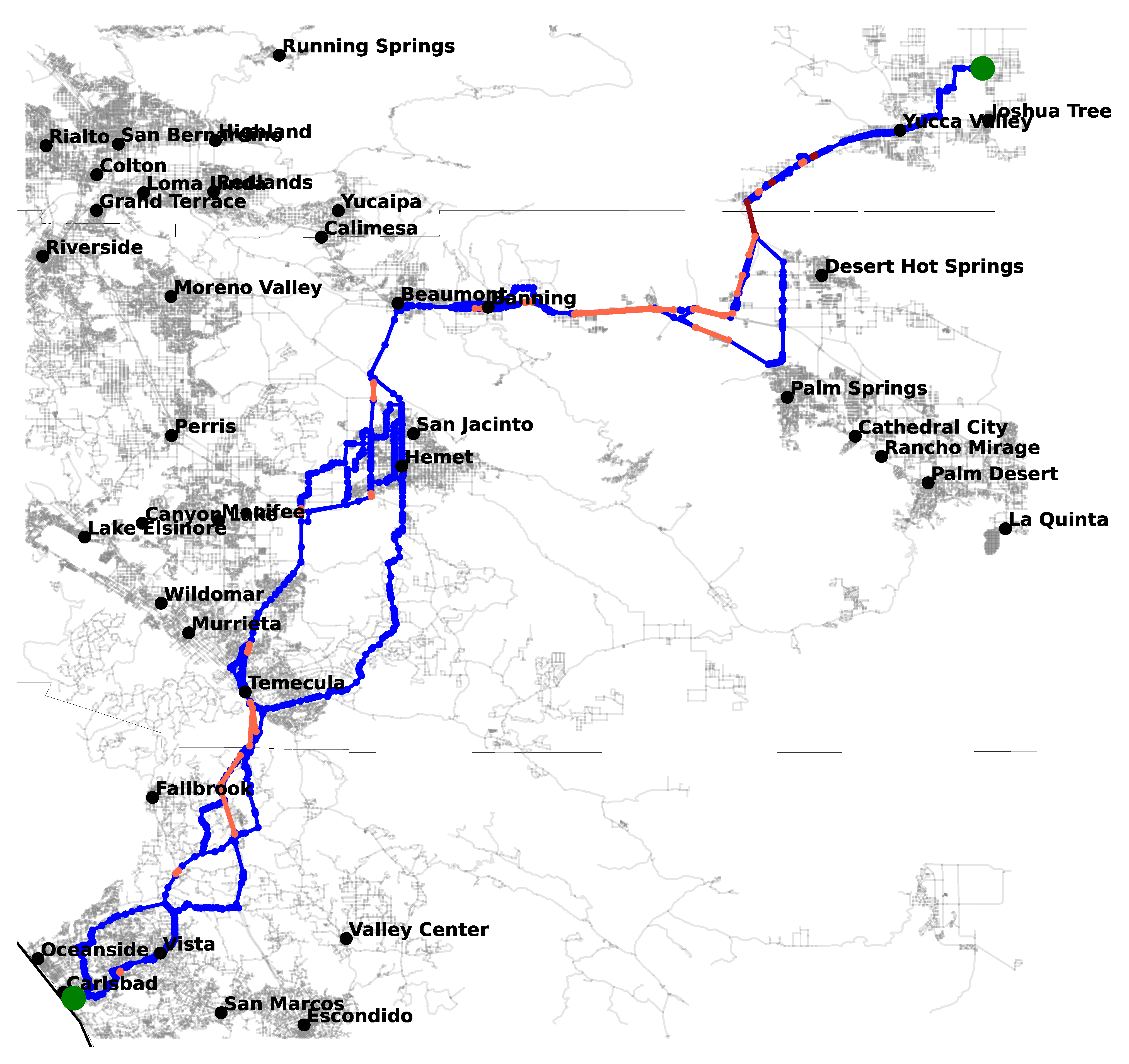}
        \caption{Red budget 4}
    \end{subfigure}

    % Row 3
    \begin{subfigure}[b]{0.48\textwidth}
        \includegraphics[width=\linewidth]{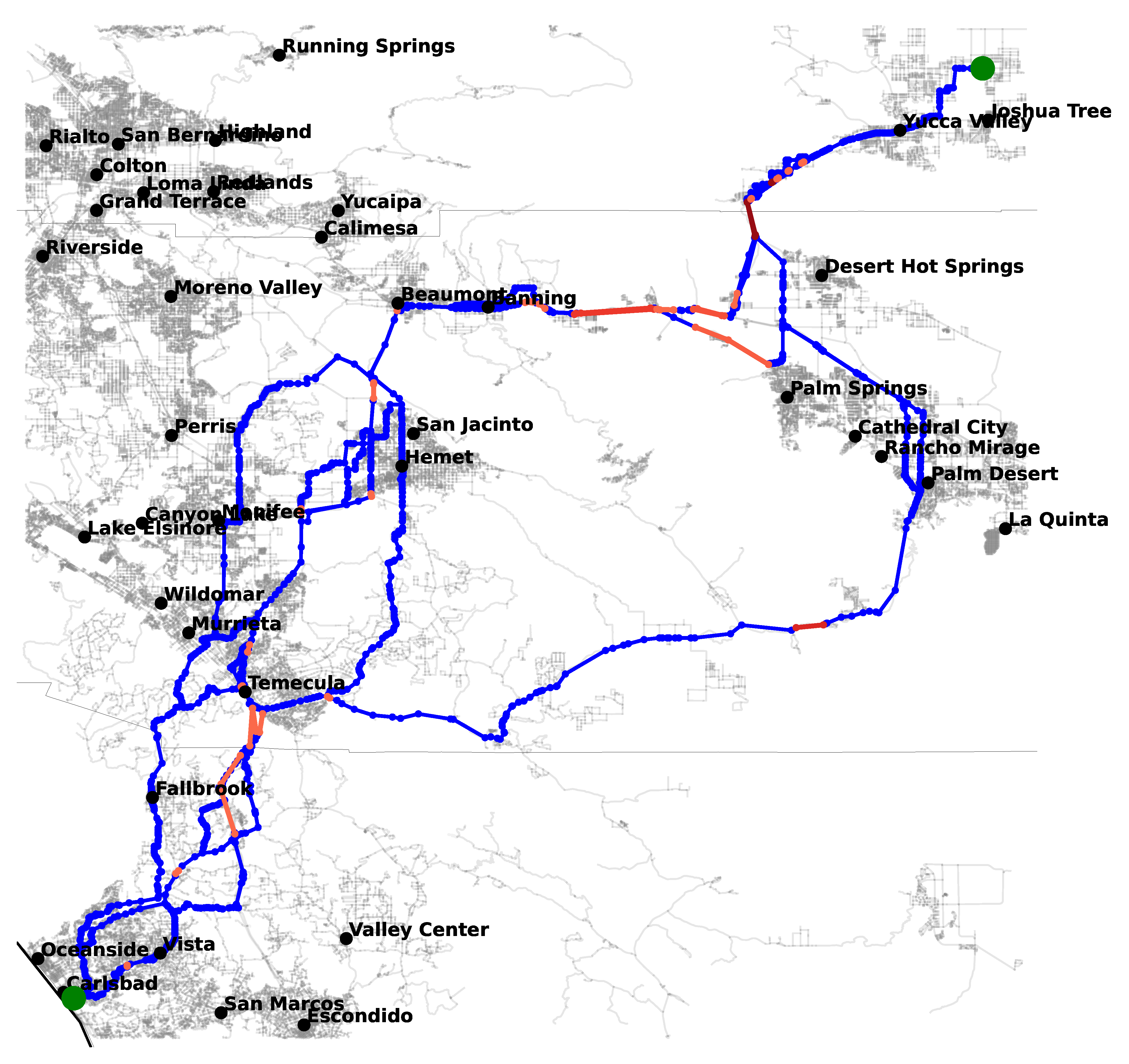}
        \caption{Red budget 5}
    \end{subfigure}
    \hfill
    \begin{subfigure}[b]{0.48\textwidth}
        \includegraphics[width=\linewidth]{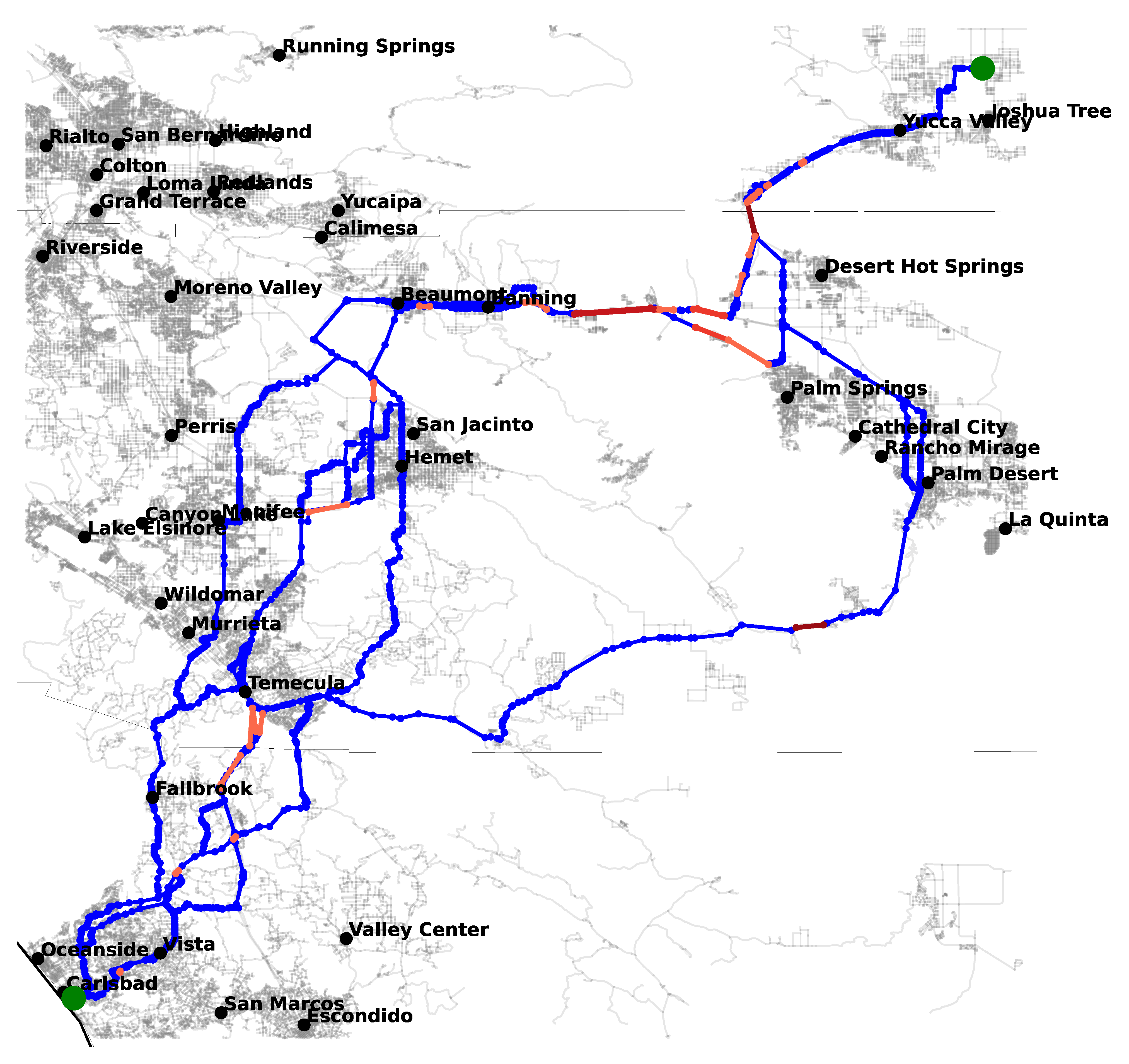}
        \caption{Red budget 6}
    \end{subfigure}

    \caption{Computed game-theoretic solutions in Scenario 2.}
    \label{fig:sc2-gt}
\end{figure}

Nevertheless, the results are still much better than those from deterministic baselines. In the fastest route, depicted on the right in Figure~\ref{fig:sc2-fp}, 34 out of 370 edges are classified as high-p-k. As a result, with a budget of 6, Red can deploy 6 high-p-k attacks, reducing throughput to 1.56\%. There is no Red-aware deterministic path that can improve upon this result. In contrast, the corresponding game-theoretic solution achieves a throughput of 29\%. Table~\ref{tab:sc2-robust} shows the robustness analysis of the computed solutions, which, with a few exceptions, remain relatively robust even in the presence of the bottlenecks.

\begin{table}[h!]
\centering
\caption{Probabilities of low (L) and high (H) p-k encounters in Scenario 2.}
\begin{tabular}{c||c|c|c|c|c|c|c|c}
~~Red budget~~ & ~no attack~ & ~~1H~~ & ~1L/1H~ & 2H & 2L/1H & 1L/2H & 3H & 1L/3H  \\
\hline
1 & 0.5 & 0.5 & - & - & - & - & - & - \\
2 & <0.01 & 0.99 & - & <0.01 & - & - & - & - \\
3 & <0.01 & <0.01 & 0.98 & <0.01 & - & - & - & - \\
4 & - & <0.01 & <0.01 & <0.01 & 0.99 & <0.01 & <0.01 & - \\
5 & - & - & 0.56 & - & <0.01 & 0.38 & - & 0.06 \\
6 & - & 0.08 & <0.01 & 0.83 & <0.01 & <0.01 & 0.08 & <0.01
\end{tabular}
\label{tab:sc2-pk}
\end{table}

\begin{table}[h!]
\centering
\caption{Throughput of solutions against different Red models in Scenario 2.}
\begin{tabular}{c||c|c|c|c|c|c}
Exp.$\backslash$True budget & ~~~~1~~~~ & ~~~~2~~~~ & ~~~~3~~~~ & ~~~~4~~~~ & ~~~~5~~~~ & ~~~~6~~~~ \\
\hline
1 & 0.75 & 0.54 & 0.41 & 0.35 & 0.27 & 0.26 \\
2 & 0.75 & 0.58 & 0.50 & 0.32 & 0.23 & 0.19 \\
3 & 0.75 & 0.53 & 0.38 & 0.28 & 0.19 & 0.16 \\
4 &  0.67 & 0.50 & 0.35 & 0.34 & 0.28 & 0.26 \\
5 & 0.67 & 0.50 & 0.31 & 0.25 & 0.27 & 0.21  \\
6 & 0.64 & 0.46 & 0.35 & 0.26 & 0.26 & 0.29
\end{tabular}
\label{tab:sc2-robust}
\end{table}

\subsection{Scenario 3: Ukraine --  Kramatorsk to Niu-York}

\begin{figure}[t]
    \centering

    % First image
    \begin{minipage}[b]{0.42\textwidth}
        \centering
        \includegraphics[draft=false,width=\linewidth]{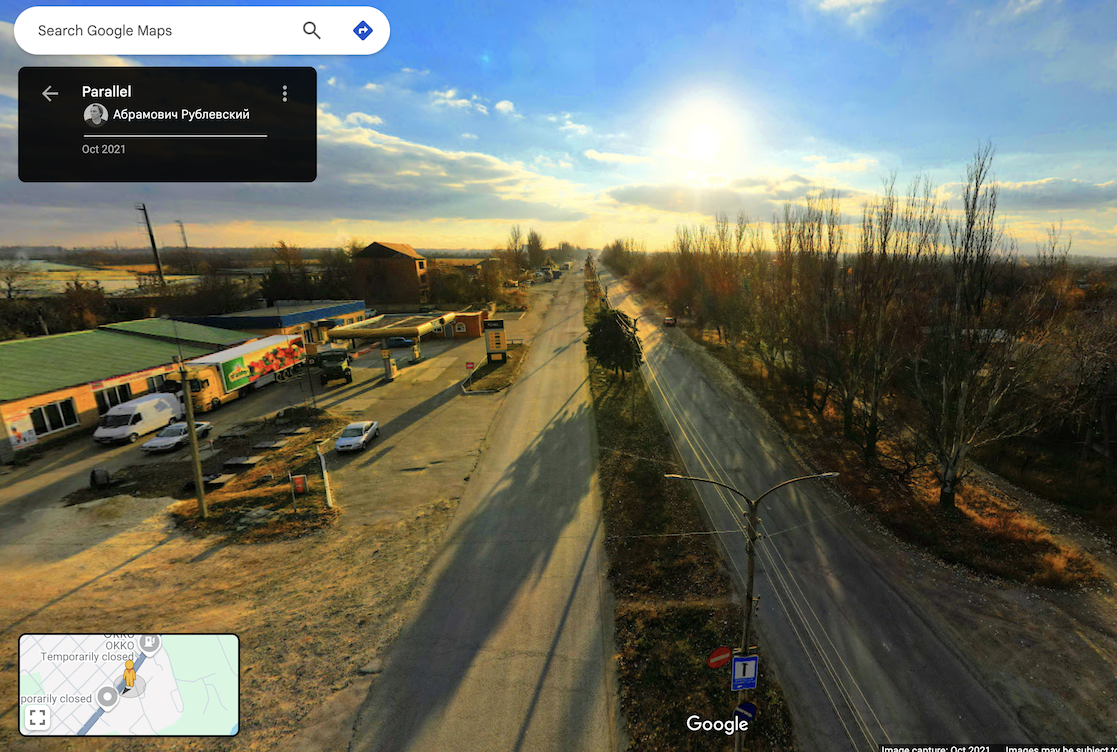}
    \end{minipage}
    \hfill
    % Second image
    \begin{minipage}[b]{0.28\textwidth}
        \centering
        \includegraphics[width=\linewidth]{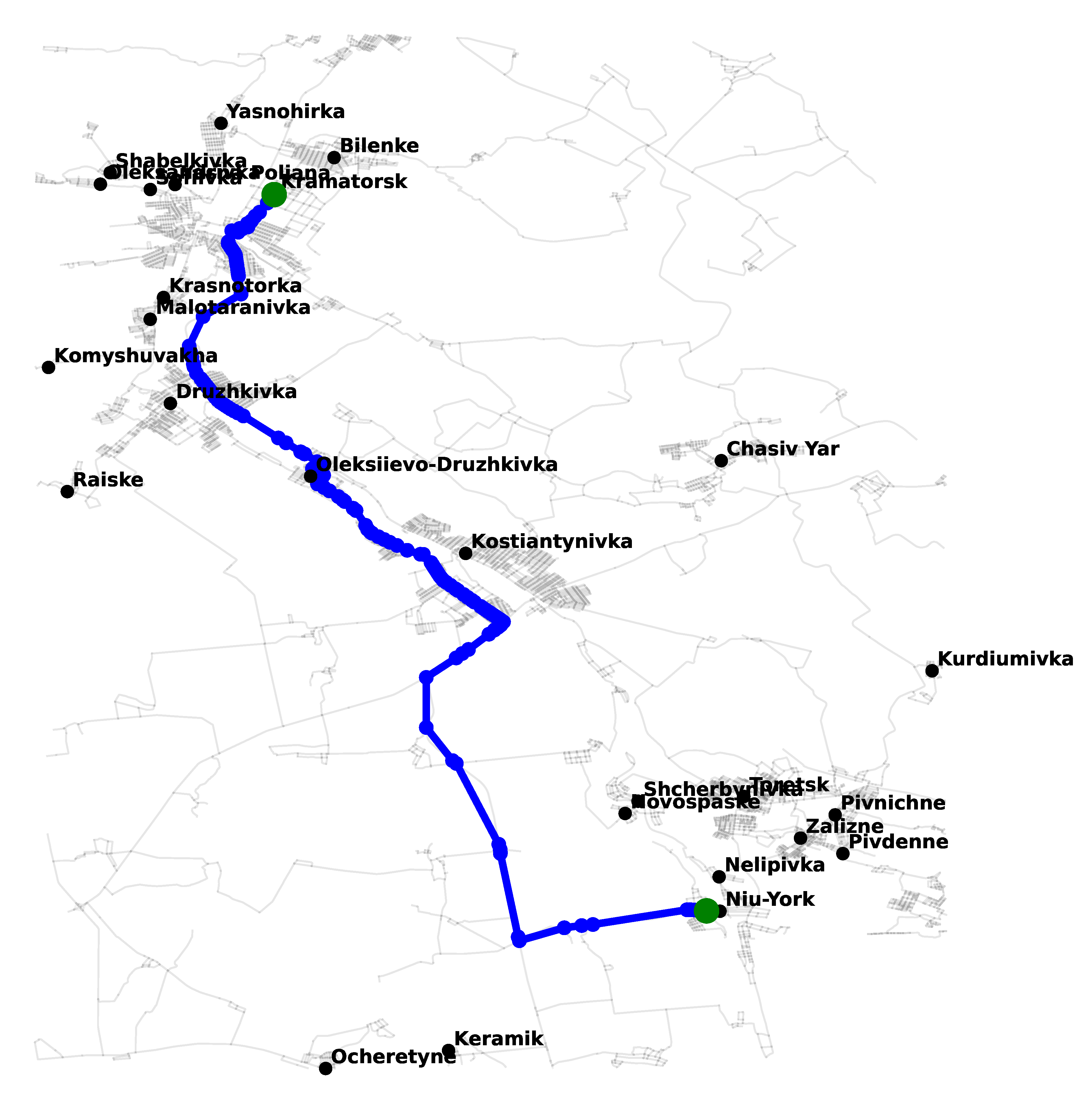}
    \end{minipage}
    \hfill
    % Third image
    \begin{minipage}[b]{0.28\textwidth}
        \centering
        \includegraphics[width=\linewidth]{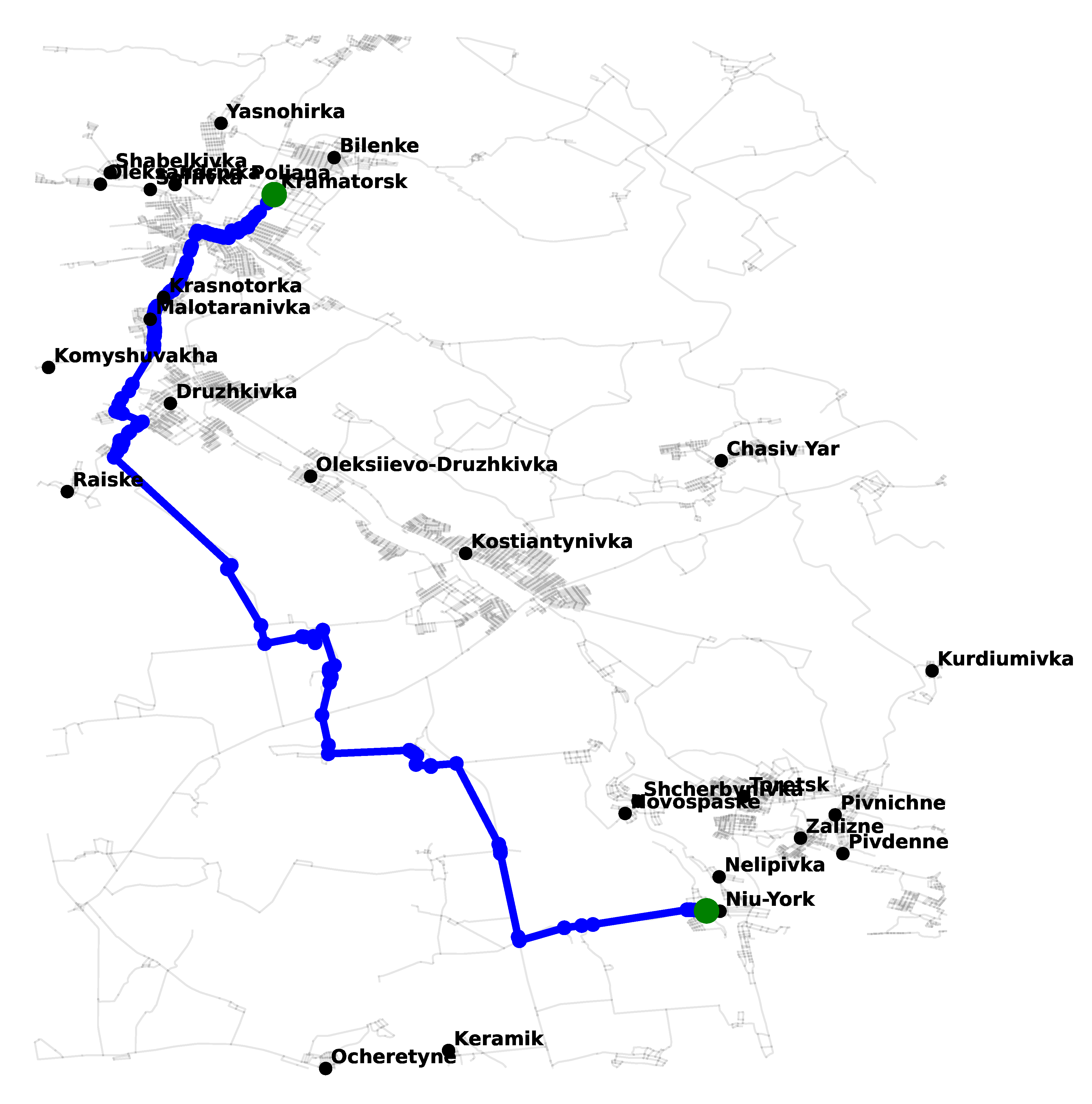}
    \end{minipage}

    \caption{Typical scenery and fastest and red-aware routes in Scenario 3.}
    \label{fig:sc3-fp}
\end{figure}

The third scenario models supply runs along the Ukrainian frontline. The OpenStreetMap graph for this region contains 5,730 nodes and 18,110 edges. The area consists of smaller towns and villages, but with a higher network density than in the second scenario. An overview of the region is provided on the left in Figure~\ref{fig:sc3-fp}. The corresponding game-theoretic solutions are presented in Figure~\ref{fig:sc3-gt}. As in the previous scenario, the size of Blue’s path support increases with Red’s budget, showing more diverse strategies of reaching the release point as Red gets stronger. Table~\ref{tab:sc3-pk} further illustrates the advantages of diversified, strategic route planning in areas with reasonably well-developed road networks. Across all budgets, in more than half of the cases, Blue is able to avoid Red entirely, completing the run unscathed. Interestingly, even with higher budgets, Red is unable to deploy more than one low and one high p-k attack.
\begin{table}[h!]
\centering
\caption{Probabilities of low (L) and high (H) p-k encounters in Scenario 3.}
\begin{tabular}{c||c|c|c|c|c|c|c|c|c}
~~Red budget~~ & no attack & 1L & 1H & 2L & 1L/1H & 2H & 3L & 2L/1H & 4L \\
\hline
1 & 0.75 & <0.01 & 0.25 & - & - & - & - & - & - \\
2 & 0.58 & 0.07 & 0.33 & - & - & 0.02 & - & - & - \\
3 & 0.56 & 0.15 & 0.21 & - & 0.07 & 0.01 & <0.01 & - & - \\
4 & 0.55 & 0.11 & 0.14 & <0.01 & 0.12 & 0.01 & 0.04 & <0.01 & 0.02 \\
5 & 0.49 & 0.10 & 0.23 & 0.02 & 0.09 & 0.01 & 0.03 & 0.02 & - \\
6 & 0.51 & 0.06 & 0.15 & 0.01 & 0.10 & 0.04 & 0.07 & <0.01 & 0.05
\end{tabular}
\label{tab:sc3-pk}
\end{table}

\begin{figure}[H]
    \centering

    % First row
    \begin{subfigure}[b]{0.32\textwidth}
        \includegraphics[width=\linewidth]{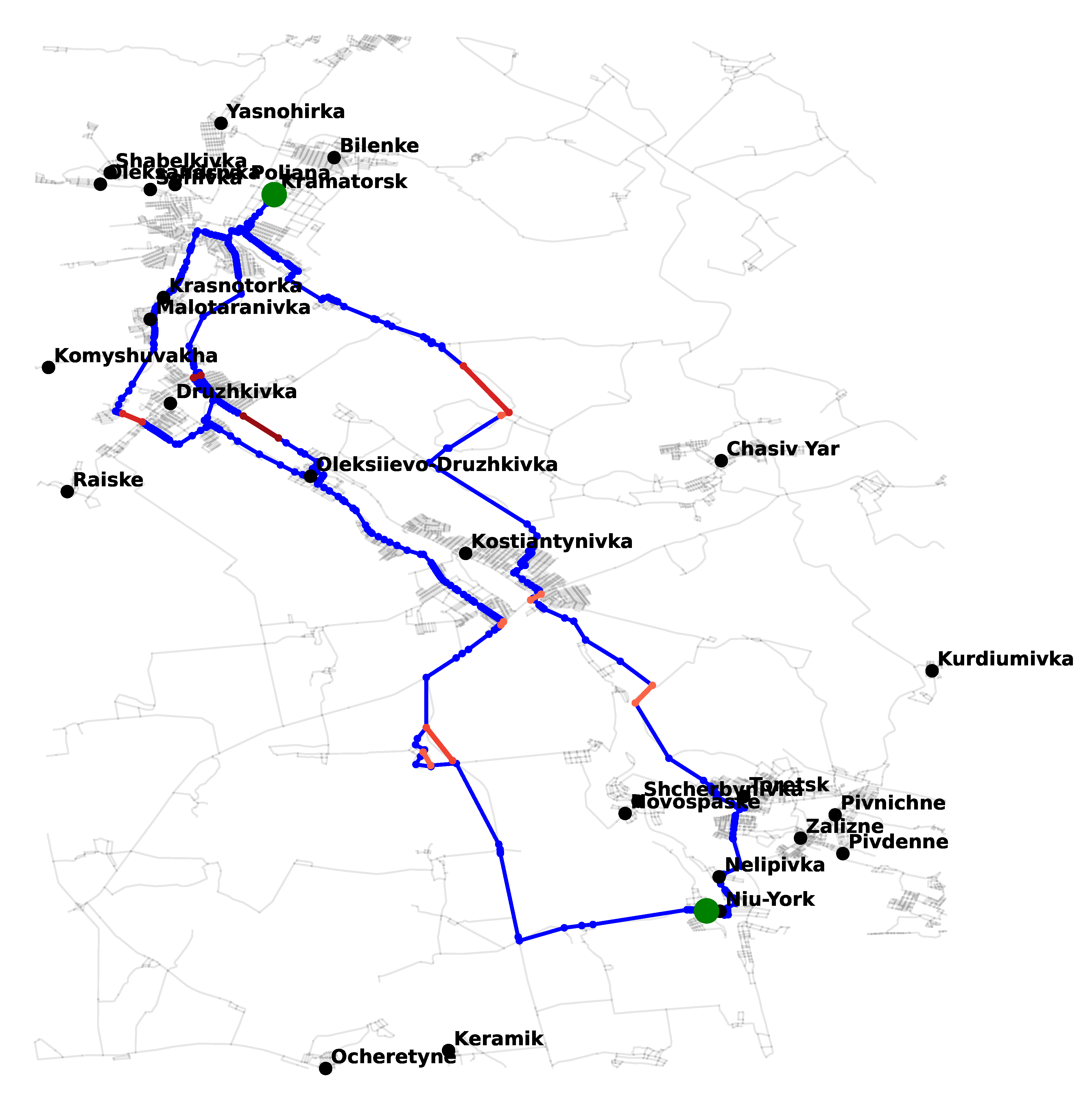}
        \caption{Red budget 1}
    \end{subfigure}
    \hfill
    \begin{subfigure}[b]{0.32\textwidth}
        \includegraphics[width=\linewidth]{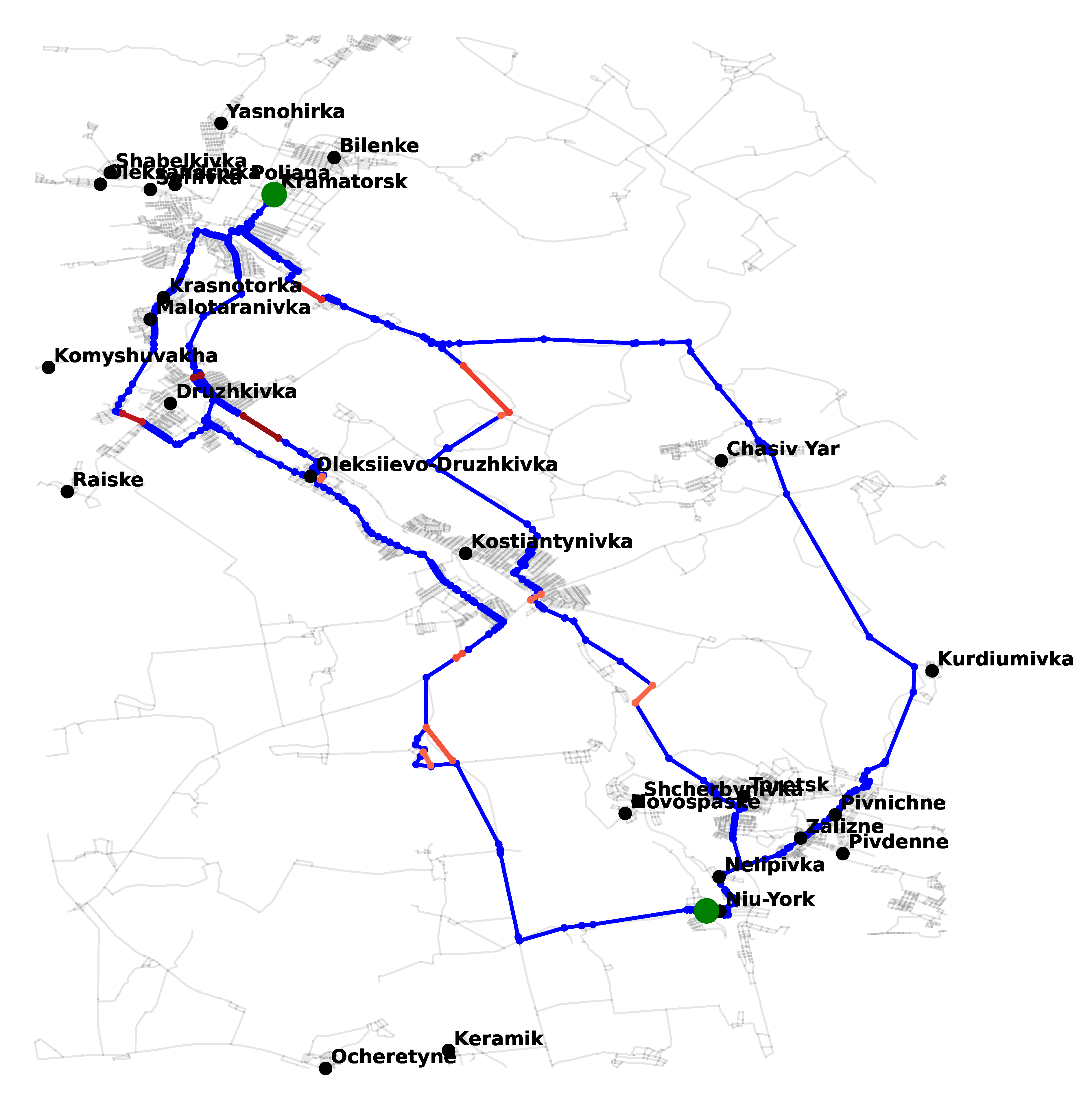}
        \caption{Red budget 2}
    \end{subfigure}
    \hfill
    \begin{subfigure}[b]{0.32\textwidth}
        \includegraphics[width=\linewidth]{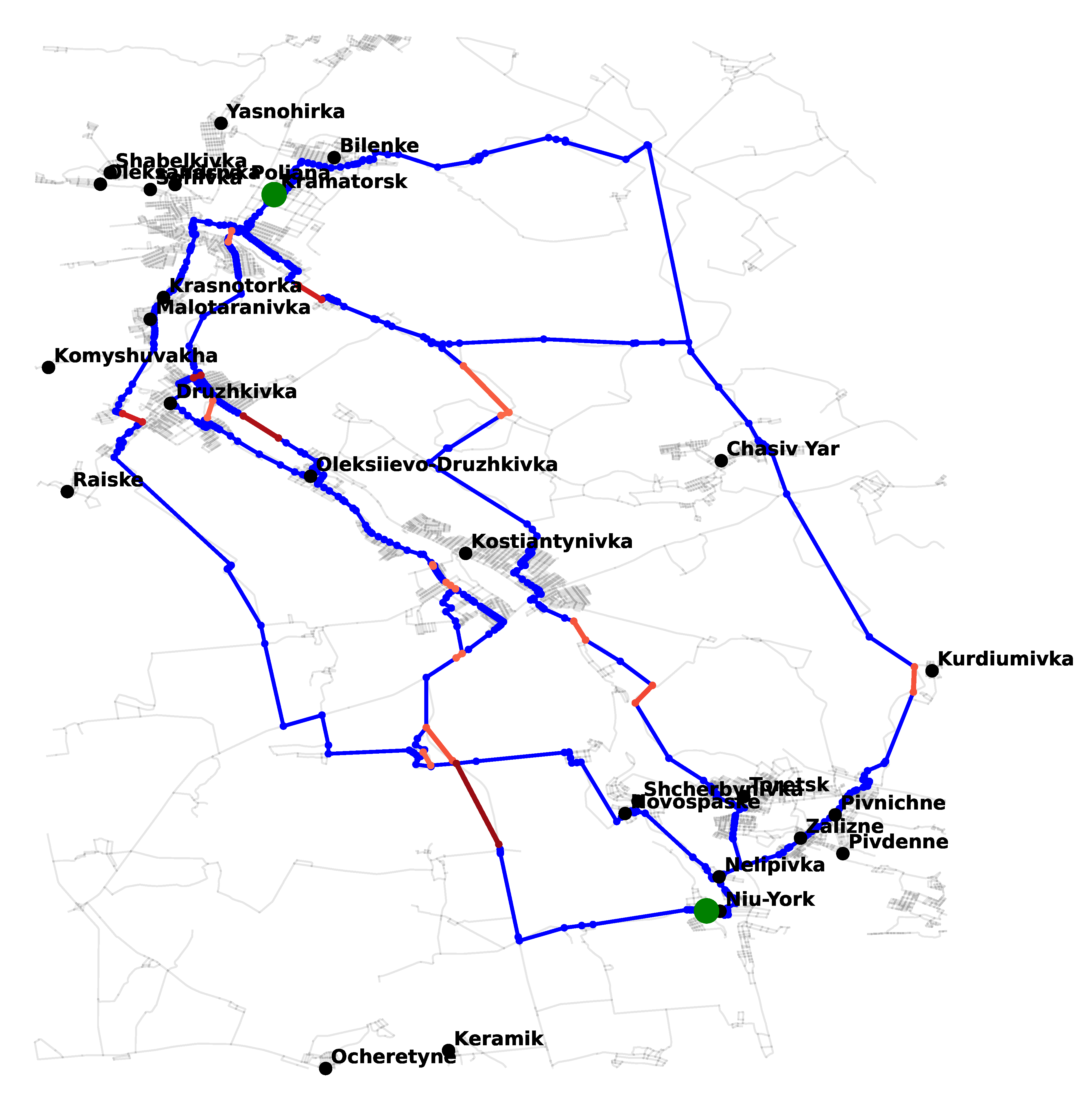}
        \caption{Red budget 3}
    \end{subfigure}

    % Second row
    \begin{subfigure}[b]{0.32\textwidth}
        \includegraphics[width=\linewidth]{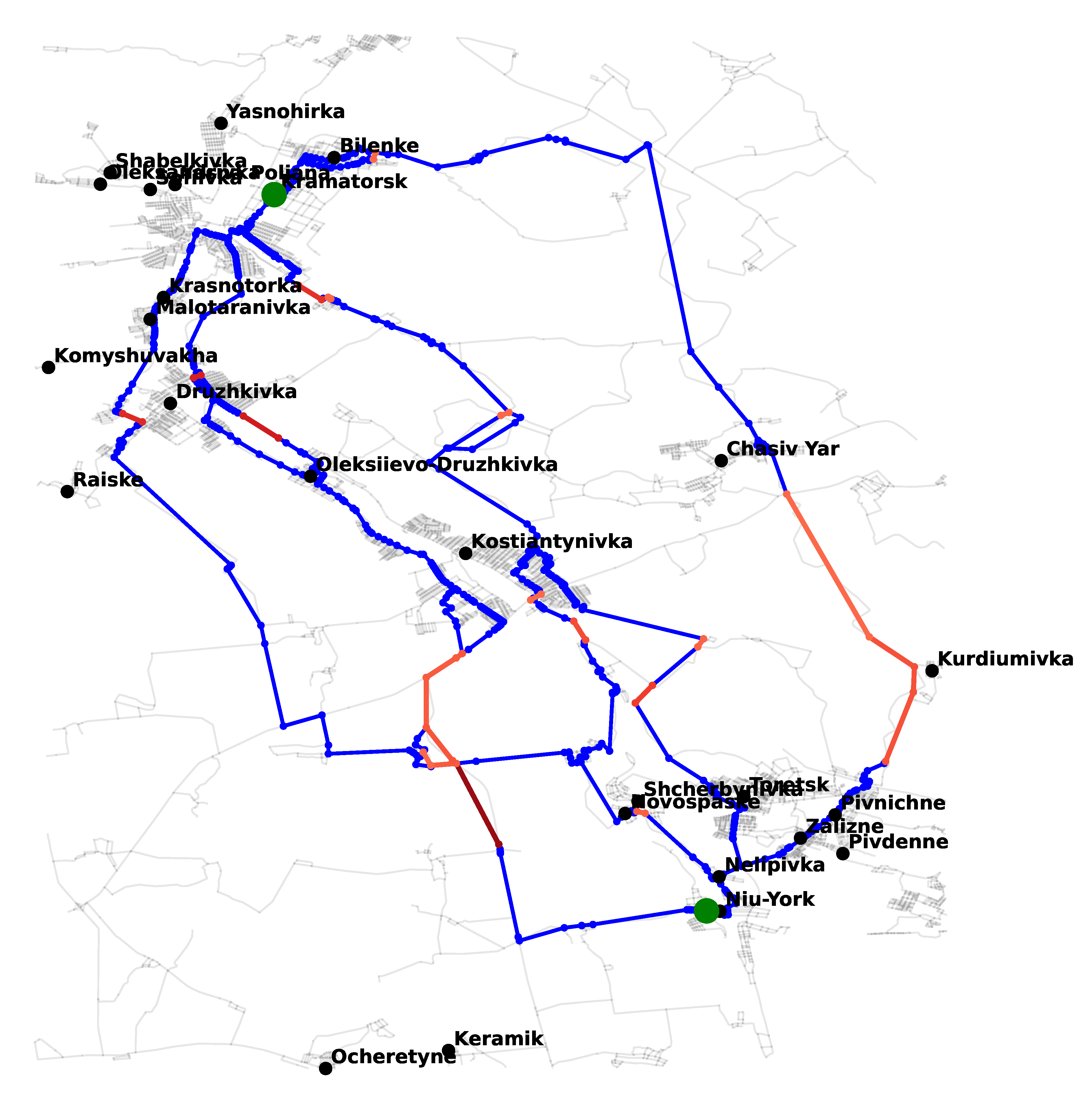}
        \caption{Red budget 4}
    \end{subfigure}
    \hfill
    \begin{subfigure}[b]{0.32\textwidth}
        \includegraphics[width=\linewidth]{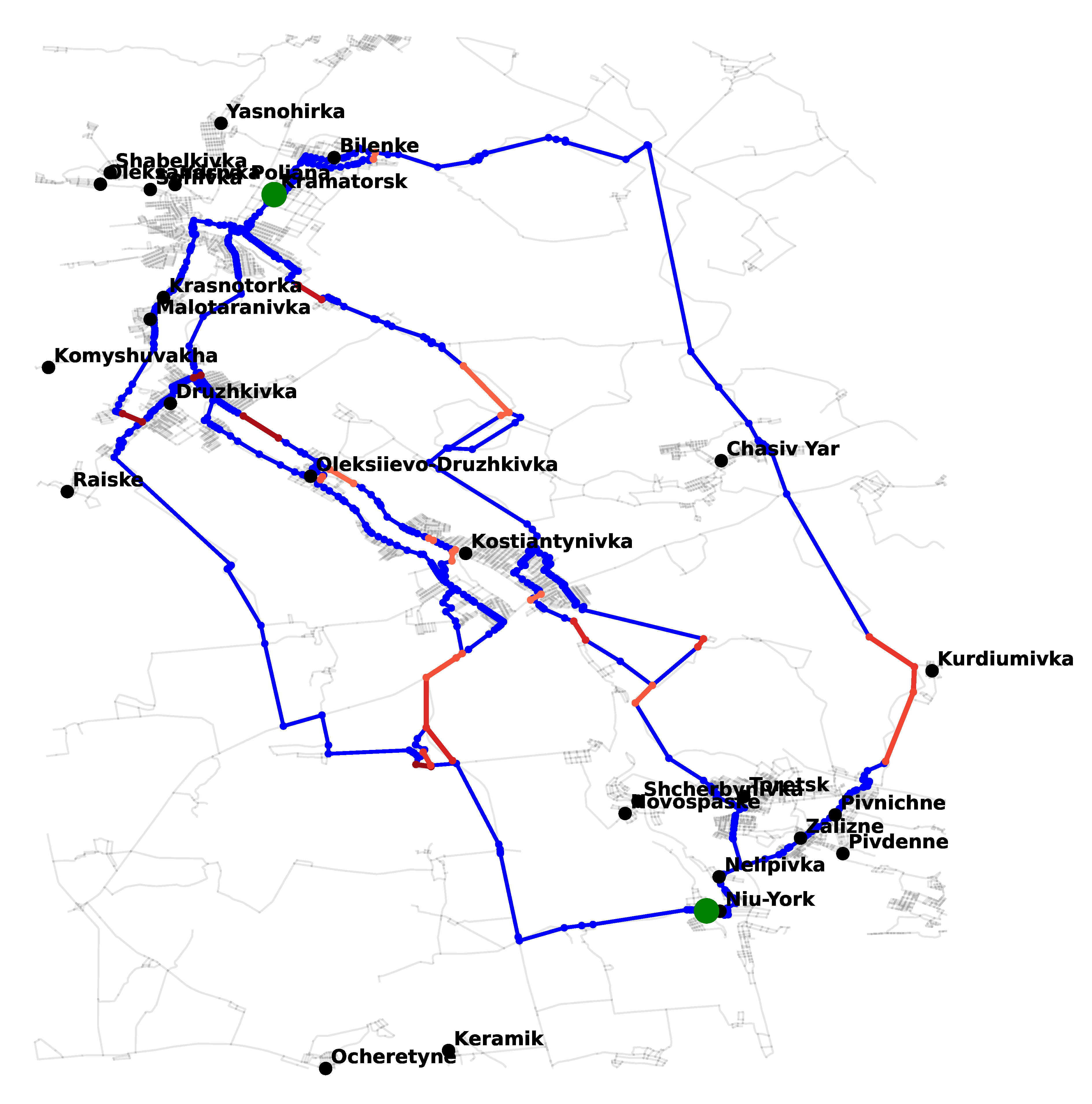}
        \caption{Red budget 5}
    \end{subfigure}
    \hfill
    \begin{subfigure}[b]{0.32\textwidth}
        \includegraphics[width=\linewidth]{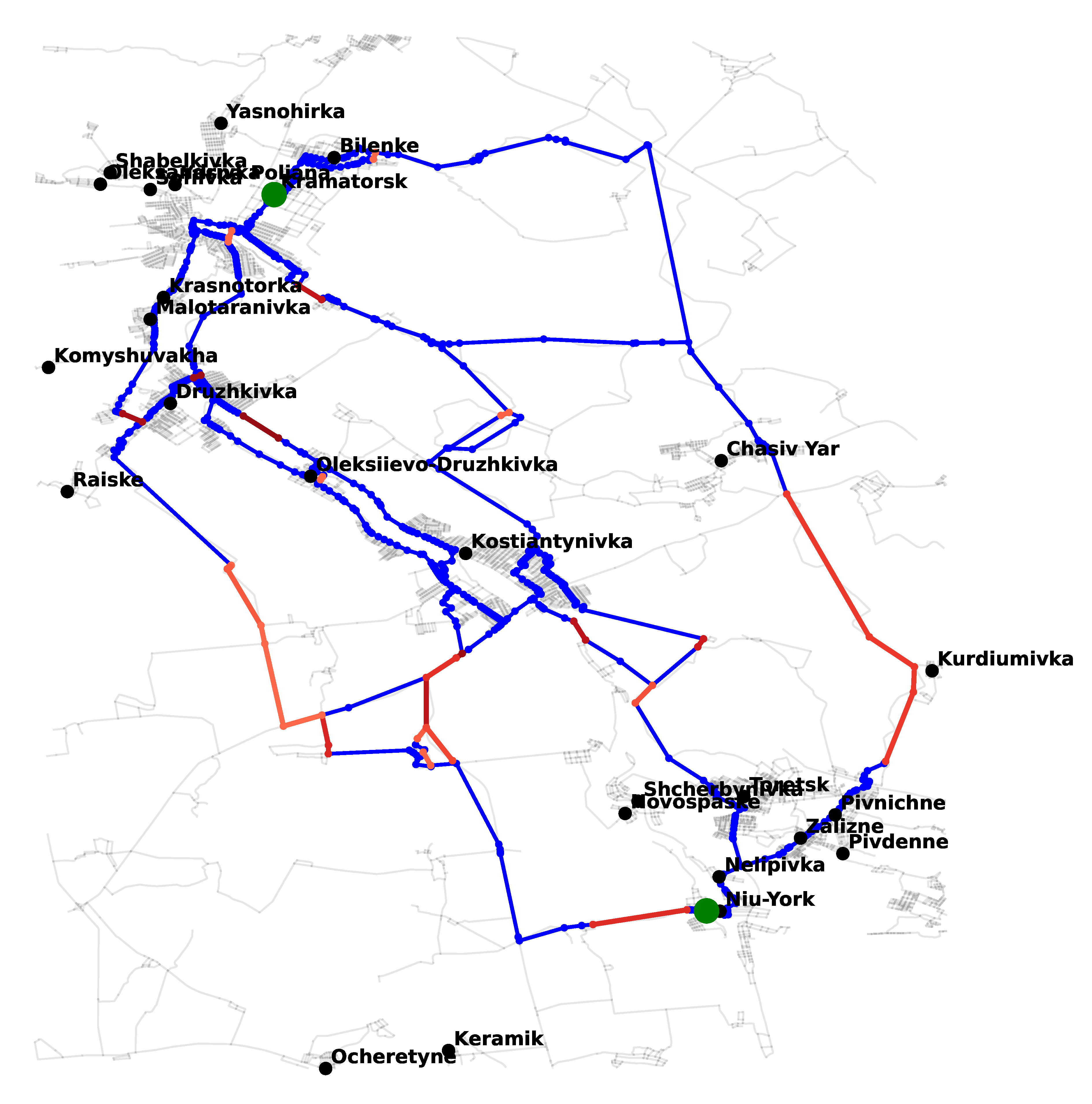}
        \caption{Red budget 6}
    \end{subfigure}
    \caption{Computed game-theoretic solutions in Scenario 3.}
    \label{fig:sc3-gt}
\end{figure}

\begin{table}[h!]
\centering
\caption{Throughput of solutions against different Red models in Scenario 3.}
\begin{tabular}{c||c|c|c|c|c|c}
Exp.$\backslash$True budget & ~~~~1~~~~ & ~~~~2~~~~ & ~~~~3~~~~ & ~~~~4~~~~ & ~~~~5~~~~ & ~~~~6~~~~ \\
\hline
1 & 0.88 & 0.81 & 0.78 & 0.63 & 0.51 & 0.42 \\
2 & 0.90 & 0.85 & 0.76 & 0.66 & 0.60 & 0.56 \\
3 & 0.88 & 0.81 & 0.89 & 0.76 & 0.73 & 0.69 \\
4 & 0.89 & 0.83 & 0.75 & 0.81 & 0.79 & 0.70 \\
5 & 0.90 & 0.80 & 0.75 & 0.66 & 0.83 & 0.80  \\
6 & 0.90 & 0.81 & 0.74 & 0.65 & 0.71 & 0.71
\end{tabular}
\label{tab:sc3-robust}
\end{table}

In comparison, the fastest deterministic route consists of 155 edges, of which 4 are high p-k. With a budget of 6, Red can carry out 2 low p-k and 3 high p-k attacks on this path, reducing throughput to 8\%. The predictable Red-aware path is subject to 4 low p-k and 1 high p-k attack, resulting in a throughput of 20.48\%. These deterministic paths are shown on the right in Figure~\ref{fig:sc3-fp}. The throughput performance of the game-theoretic solutions under Red models with different budgets is summarized in Table~\ref{tab:sc3-robust}. Against the correct Red model, throughput remains above 70\%. However, in this scenario, we observe a decrease in robustness, highlighting the importance of accurate Red modeling. In particular, with higher expected and actual budgets, throughput varies more substantially, though it still significantly outperforms the deterministic alternatives.

\section{Conclusion}

In this paper, we introduced a contested route planning framework that models adversarial interdiction in routing scenarios. The problem is formulated as a two-player zero-sum game on a graph, with optimal route strategies identified via randomized Nash equilibria. We demonstrated the computational complexity of solving for these equilibria and presented a practical double-oracle algorithm combining best-response oracles for both players. Our experiments on realistic urban networks confirm that the approach scales well to problems of practical size. These models, developed in cooperation with the U.S. Office of Naval Research, are intended for practical use in contested environments. The experimental results underscore the importance of explicitly modeling adversarial behavior rather than relying on deterministic routes.

\begin{credits}
\subsubsection{\ackname} 
This research was supported by the Office of Naval Research awards N00014-22-1-2530 and N00014-23-1-2374, and the National Science Foundation awards IIS-2147361 and IIS-2238960. We would like to thank Sven Bedn\'{a}\v{r} for his assistance with formalizing the Ukraine scenario.
% \subsubsection{\discintname}
\end{credits}
\bibliographystyle{splncs04}
\bibliography{library}
\end{document}